\documentclass[journal]{IEEEtran}
\usepackage{cite}
\usepackage{amsmath,amsfonts,amscd}
\usepackage{epsfig} 
\usepackage{mathptmx} 
\DeclareMathAlphabet{\mathcal}{OMS}{cmsy}{m}{n}
\usepackage{helvet}
\usepackage{amssymb}
\usepackage{courier}
\usepackage{type1cm}  
\usepackage{bbm}
\usepackage{dblfloatfix}
\usepackage{bm}
\usepackage{array}
\usepackage{mathtools}
\usepackage{wrapfig}
\usepackage[utf8]{inputenc}
\usepackage{algorithmic}
\usepackage{algorithm}
\usepackage[T1]{fontenc}
\usepackage{url}
\usepackage{float}
\usepackage{booktabs}
\usepackage{nicefrac}
\usepackage{microtype}   
\usepackage{graphicx}
\usepackage{multirow}
\usepackage{subcaption}
\graphicspath{{Figs/}}
\usepackage{xcolor}

\usepackage{amsthm}

\newtheorem{lemma}{Lemma}
\newtheorem{theorem}{Theorem}

\newtheorem{corollary}{Corollary}
\newtheorem{proposition}{Proposition}
\newtheorem{assumption}{Assumption}
\theoremstyle{definition}
\newtheorem{definition}{Definition}
\newtheorem{remark}{Remark}

\DeclareMathOperator*{\argmin}{arg\,min}



\usepackage{tikz}
\usetikzlibrary{arrows.meta,positioning,fit,backgrounds,calc}
\usepackage{hyperref}
\usepackage{cleveref}

\usepackage{textcomp}
\def\BibTeX{{\rm B\kern-.05em{\sc i\kern-.025em b}\kern-.08em
    T\kern-.1667em\lower.7ex\hbox{E}\kern-.125emX}}

\begin{document}
\title{Safety-Critical Contextual Control via Online Riemannian Optimization with World Models}
\author{Tongxin~Li,~\IEEEmembership{Member,~IEEE}
}
\maketitle

\begin{abstract}
Modern world models are becoming too complex to admit explicit dynamical descriptions. We study \emph{safety-critical contextual control}, where a Planner must optimize a task objective using only feasibility samples from a black-box Simulator, conditioned on a context signal $\xi_t$. We develop a sample-based Penalized Predictive Control (PPC) framework grounded in online Riemannian optimization, in which the Simulator compresses the feasibility manifold into a score-based density $\hat{p}(u \mid \xi_t)$ that endows the action space with a Riemannian geometry guiding the Planner's gradient descent. The barrier curvature $\kappa(\xi_t)$, the minimum curvature of the conditional log-density $-\ln\hat{p}(\cdot\mid\xi_t)$, governs both convergence rate and safety margin, replacing the Lipschitz constant of the unknown dynamics. Our main result (Theorem~\ref{thm:iss_rigorous}) is a contextual safety bound showing that the distance from the true feasibility manifold is controlled by the score estimation error and a ratio that depends on $\kappa(\xi_t)$, both of which improve with richer context. Simulations on a dynamic navigation task confirm that contextual PPC substantially outperforms marginal and frozen density models, with the advantage growing after environment shifts.
\end{abstract}

\section{Introduction}
\label{sec:intro}

\IEEEPARstart{T}{he} increasing complexity of modern world models is fundamentally changing the interface between planning and control. Neural physics simulators~\cite{sanchez2020learning}, video-generation models~\cite{yang2024video}, and foundation-model-based world models~\cite{ha2018world,bruce2024genie} can predict the consequences of actions with remarkable fidelity, yet their internal representations (multi-modal embeddings, high-dimensional latent states, or raw image sequences) are far too complex to admit explicit and known dynamical descriptions of the form $x_{t+1} = f_t(x_t, u_t)$. These black-box world models offer a compelling advantage over classical model-based control because they capture richer dynamics than any tractable closed-form approximation, generalize across environments without manual re-derivation, and directly leverage the rapid advances in foundation models and learned simulators. However, a controller that must guarantee safety while interacting with such a \textit{black-box world model} faces a fundamental challenge. It cannot write down the dynamics, cannot form the constraint Jacobians that classical Model Predictive Control (MPC) requires, and cannot construct the explicit barrier functions used in safety-critical control (\Cref{fig:paradigm}).

\begin{figure*}[t]
\centering
\begin{tikzpicture}[
    block/.style={draw, rounded corners, minimum height=1.05cm, minimum width=2.0cm, align=center, font=\footnotesize},
    arr/.style={-{Stealth[length=2.5mm]}, thick},
    lbl/.style={font=\scriptsize, fill=white, inner sep=1.5pt},
]
\begin{scope}[local bounding box=classical]
\node[block, fill=green!6] (envA) {Environment\\$x_{t+1}{=}f_t(x_t,u_t)$};
\node[block, right=1.2cm of envA] (modA) {Explicit\\Model $f_t$};
\node[block, right=1.0cm of modA] (jacA) {Jacobians\\$\nabla_u f_t$};
\node[block, right=1.0cm of jacA] (ctrlA) {Controller\\$\min\, c(u)$};
\node[block, right=1.0cm of ctrlA] (qpA) {CBF-QP\\Filter};
\draw[arr] (envA) -- node[lbl,above]{$x_t$} (modA);
\draw[arr] (modA) -- (jacA);
\draw[arr] (jacA) -- node[lbl,above]{$\nabla h$} (ctrlA);
\draw[arr] (ctrlA) -- (qpA);
\draw[arr, rounded corners=5pt] (qpA.south) -- ++(0,-0.45) -| node[lbl,below,pos=0.22]{$u_t$} (envA.south);
\end{scope}
\node[font=\footnotesize\bfseries, anchor=east] at ([xshift=-6pt]classical.west) {(a)};

\begin{scope}[yshift=-2.4cm, local bounding box=ours]
\node[block, fill=green!6] (envB) {Environment};
\node[block, fill=blue!8, right=1.2cm of envB] (simB) {Black-Box\\Simulator};
\node[block, fill=blue!8, right=1.0cm of simB] (kdeB) {KDE\\$\hat{p}_t,\,\hat{s}_t$};
\node[block, fill=orange!10, right=1.0cm of kdeB] (planB) {Planner\\$\min\, \mathcal{F}(u)$};
\node[block, fill=blue!8, right=1.0cm of planB] (filtB) {Safety\\Filter};
\draw[arr] (envB) -- node[lbl,above]{$x_t$} (simB);
\draw[arr] (simB) -- node[lbl,above]{samples} (kdeB);
\draw[arr] (kdeB) -- node[lbl,above]{$\hat{s}_t,\,\beta_t$} (planB);
\draw[arr] (planB) -- node[lbl,above]{$u_{\mathrm{ppc}}$} (filtB);
\draw[arr, rounded corners=5pt] (filtB.south) -- ++(0,-0.45) -| node[lbl,below,pos=0.22]{$u_t$} (envB.south);
\draw[arr, dashed, rounded corners=4pt] (envB.north) -- ++(0,0.5) -| node[font=\small, fill=white, inner sep=2pt, pos=0.3]{context $\xi_t$} (planB.north);
\end{scope}
\node[font=\footnotesize\bfseries, anchor=east] at ([xshift=-6pt]ours.west) {(b)};
\end{tikzpicture}
\caption{Two closed-loop control paradigms. \textbf{(a)}~Classical model-based control requires an explicit dynamics model $f_t$ to form constraint Jacobians and a CBF-QP safety filter that projects the controller's action onto the safe set. \textbf{(b)}~The contextual control framework replaces the explicit model with a black-box Simulator that produces feasibility samples; a KDE compresses them into a score signal $\hat{s}_t$, which the Planner uses for gradient-based optimization. A safety filter (score-ascent retraction) projects the candidate $u_{\mathrm{ppc}}$ onto $\hat{\mathcal{M}}_t^\alpha$ if needed. The Planner also receives a context signal $\xi_t$ from the environment (dashed arrow), enabling conditional density estimation $\hat{p}(u \mid \xi_t)$. See~\Cref{fig:alg_flow} for the detailed per-step information flow.}
\label{fig:paradigm}
\end{figure*}
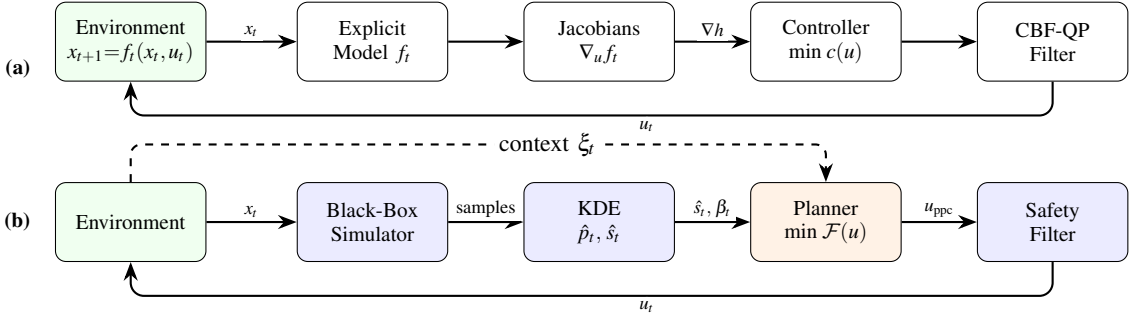

We address this challenge by developing a framework for safe control through black-box world-model samples. The controller optimizes an objective over a feasibility set whose geometry is accessible only through \textit{feasibility samples} drawn from the world model, rather than through an explicit dynamics model. The world model encapsulates all complexity of the physical environment, and the controller's access to safety-relevant information is mediated entirely by a learned \textit{density signal} that compresses the simulator's knowledge of the feasibility manifold. This architecture separates \textit{what is feasible} (determined by the black-box world model) from \textit{what is desirable} (the controller's objective), and the score function of the learned density provides a differentiable interface between the two.

This separation is not merely an engineering convenience; it is an inevitability as world models become richer and more expressive. The architecture arises naturally in robotics with visual world models~\cite{ha2018world,guo2026ctrlworld}, where a planner cannot extract closed-form obstacle constraints from a neural simulator's latent state. It appears equally in multi-agent energy dispatch~\cite{li2021information,li2021learning}, where an aggregator coordinates distributed resources whose local feasibility sets are private and time-varying, so the aggregator must plan using only sampled feasibility signals rather than explicit local models, and in autonomous driving~\cite{rhinehart2019precog}, where a predictive model simulates traffic scenarios but its internal state is not available to the planner in closed form. In each scenario, the feasibility manifold $\mathcal{M}_t$ is a time-varying, non-convex subset of the decision space whose geometry is determined by the world model's internal state. The controller must optimize over $\mathcal{M}_t$ \textit{without} an explicit dynamics model, learning the manifold structure from online samples while guaranteeing safety at every step.

This paper addresses the following question. \textit{Can a controller guarantee safety when its only access to the constraint geometry is through black-box world-model samples?}

We answer affirmatively. Our approach treats the world model as a \textit{Simulator} that compresses the geometry of the time-varying feasibility manifold into a score-based density signal, which a \textit{Planner} then uses for safe gradient-based optimization. When the Planner's observation is a \emph{context} $\xi_t$, a fixed-dimensional embedding of rich sensory data rather than the Simulator's full state, we estimate a \emph{conditional} feasibility density $\hat{p}(u \mid \xi_t)$ so that the learned manifold tracks changes in the environment that are reflected in $\xi_t$. This perspective builds on classical connections between control and probabilistic inference~\cite{todorov2006linearly,kappen2005linear} and extends recent score-based manifold optimization~\cite{kharitenko2025landing} to the online, safety-critical setting.

\subsubsection{Model-Free and Simulator-Based Control}
Classical MPC~\cite{garcia1989model} and its robust variants~\cite{bemporad1999robust,mayne2005robust} require an explicit dynamics model to construct constraints and compute Jacobians. Data-driven MPC~\cite{berberich2020data,bold2024data,rosolia2018data} learns surrogate models but still assumes the learned model is amenable to explicit constraint formulation. Recent work on predictive safety filters~\cite{wabersich2021predictive} and robust adaptive NMPC~\cite{buerger2026robust} further improves constraint handling, yet remains tied to an explicit or partially identified plant model. The Cross-Entropy Method (CEM)~\cite{rubinstein1999cross} is a popular sampling-based alternative that draws candidate actions from a parametric proposal, evaluates them against a model, and refits the proposal to an elite subset. CEM-MPC variants~\cite{chua2018deep,pinneri2021sample} combine CEM with learned dynamics models, but provide no safety certificates and scale poorly with action dimension. The control-as-inference viewpoint~\cite{todorov2006linearly,kappen2005linear} and rate-distortion theory~\cite{tishby2000information} provide a principled way to trade off performance and information cost, but existing instantiations assume known dynamics. Our framework makes no such assumption. The world model is accessed purely through feasibility samples, and the information-theoretic penalty emerges as the Lagrange multiplier of the sampling capacity, providing a mechanism for adaptive trust calibration unavailable in model-based or sampling-based approaches.

\subsubsection{Riemannian Optimization and Density Models}
Score matching~\cite{hyvarinen2005score,vincent2011connection} estimates the gradient of the log-density without computing the normalizing constant. Classical Riemannian optimization~\cite{boumal2023introduction,absil2008optimization} requires explicit manifold descriptions, but recent work~\cite{kharitenko2025landing} shows that score functions of smoothed densities can recover the manifold projection and tangent-space operations needed for optimization in a static, offline regime. Normalizing Flows~\cite{rezende2015variational,kobyzev2020normalizing} have been applied to trajectory forecasting~\cite{rhinehart2019precog} and policy parameterization~\cite{chao2024maximum}, while the information-aggregation framework of~\cite{li2021information,li2021learning} encodes distributed constraint geometry into a feasibility density. None of these works addresses \emph{online} score estimation with provable safety guarantees. Our framework fills this gap by extending density-based manifold operations to the online setting where the manifold is unknown and time-varying, and provides the first safety certificates for score-based feasibility learned from black-box samples.

\subsubsection{Data-Driven Safety Certificates}
Control Barrier Functions (CBFs)~\cite{taylor2020learning,ames2019control,ames2016control} parameterize safe sets as superlevel sets $\{x : h(x) \geq 0\}$ of learned functions. Recent extensions address parametric uncertainty~\cite{pati2025robust} and enable purely data-driven CBF synthesis from input-output measurements~\cite{he2026direct,bajelani2026data}. These methods suffer from \textit{structural bias}: maintaining tractability requires restricting $h$ to specific function classes (polynomials, GPs~\cite{cheng2019end}), which under-approximate complex, non-convex geometries. GP-based approaches additionally scale poorly with state dimension due to cubic kernel-matrix inversion. Our score-based approach can represent arbitrary manifold topologies without structural assumptions on $h$, scales favorably in the control dimension, and refines its safety certificate online as samples accumulate.

\textbf{Our contributions} are three-fold.

\emph{First}, we formalize \emph{contextual control with black-box world models} as a new problem class (\Cref{sec:problem}) in which the controller has no access to the dynamics $f_t$ and receives only feasibility samples from a black-box Simulator, yet must certify safety at every step. To our knowledge, this is the first framework that unifies black-box simulator access, online density estimation, and provable safety guarantees under a single formulation.

\emph{Second}, we develop a sample-based instantiation of Penalized Predictive Control (PPC)~\cite{li2021information} for the contextual black-box setting. We derive a rate-distortion formulation (Proposition~\ref{prop:thermo}) that yields the penalty weight $\beta$ as the Lagrange multiplier of the sampling capacity, together with an online score-estimation pipeline that replaces the oracle density with a learned estimate. The resulting controller uses the \emph{score} of the learned density as its sole interface to the constraint geometry, replacing the explicit Jacobians and barrier functions of classical MPC. A central novelty is that the \emph{barrier curvature} $\kappa$, the minimum eigenvalue of the Fisher information $\mathcal{I}_{\hat{p}}$, plays the role that the dynamics Lipschitz constant plays in standard safety-critical control. Because $\kappa$ is a property of the learned density rather than the unknown dynamics, it can be estimated from data and used to set $\beta$ adaptively, which is impossible in classical barrier-function methods.

\emph{Third}, we establish a sequence of convergence and safety guarantees, each building on the preceding one. Theorem~\ref{thm:contraction} shows that PPC iterates contract geometrically at a rate set by the barrier curvature $\kappa$, and Proposition~\ref{prop:betacrit} identifies the critical stiffness
above which the equilibrium is guaranteed to lie inside the learned manifold. When the manifold drifts, Theorem~\ref{thm:dynamic_regret} bounds the dynamic regret in terms of the score variation $\mathcal{V}_T$, revealing a \emph{dual role} of $\kappa$, namely it speeds contraction and simultaneously dampens sensitivity to manifold drift. Theorem~\ref{thm:learning} then bounds the score estimation error at rate $O(N^{-2/(m+4)})$, which feeds into Theorem~\ref{thm:iss_rigorous} to yield a safety bound whose residual $G_c/(\beta\kappa)$ depends on the learned density's curvature rather than a Lipschitz constant of the unknown dynamics. Finally, Proposition~\ref{prop:ctx_curvature} derives, via a variational-inference argument, a mixture-Hessian identity showing that marginalizing out $\xi_t$ subtracts a posterior-score covariance from the learned curvature, and Theorem~\ref{thm:ctx_gap} uses this to give a formal safety-guarantee gap of at least $G_c\,\sigma_t^2/(\beta_t\,\kappa(\xi_t)\,\kappa_{\mathrm{marg}})$ between the Theorem~\ref{thm:iss_rigorous} residuals of the context-aware and context-blind Planners whenever the current context is identifiable from the action, with $\sigma_t^2$ the minimum eigenvalue of the posterior conditional-score covariance. To our knowledge, this is the first formal characterization of when and by how much context improves safety in a black-box control setting.

\section{Problem Formulation}
\label{sec:problem}
We consider an architecture composed of a \textit{Planner} responsible for optimizing a global objective and a black-box \textit{Simulator} (world model) that encapsulates the physical environment. The system evolves according to time-varying nonlinear dynamics:
\begin{equation}
\label{eq:system}
    x_{t+1} = f_t(x_t, u_t), \quad t = 0, \ldots, T,
\end{equation}
where $x_t \in \mathcal{X} \subseteq \mathbb{R}^n$ is the system state and $u_t \in \mathcal{U} \subseteq \mathbb{R}^m$ is the control input. The dynamics $f_t$ are \textit{embedded inside the Simulator} and are not available to the Planner in closed form. The time-dependence of $f_t$ captures exogenous factors (moving obstacles, changing network topology, or shifting workload patterns) that the Simulator can observe and simulate but that remain hidden from the Planner.

\textbf{Notation.} Throughout, $\|\cdot\|$ denotes the Euclidean norm and $A \preceq B$ the L\"{o}wner order on symmetric matrices. We write $\mathrm{dist}(u, S) = \inf_{v \in S}\|u - v\|$ for point-to-set distance and $d_H$ for Hausdorff distance. The following symbols recur across sections and are collected here for reference: $T$ is the planning horizon; $K$ is the number of inner gradient steps per PPC update; $N_t$ is the cumulative number of feasibility samples available at time $t$; $\hat{p}_t$ is the learned feasibility density and $\hat{s}_t = \nabla_u \ln \hat{p}_t$ its score function; $\hat{\mathcal{M}}_t^\alpha = \{u : \hat{p}_t(u \mid \xi_t) \geq \alpha\}$ is the learned $\alpha$-superlevel set at time $t$; $\mathcal{I}_{\hat{p}}(u) = -\nabla^2 \ln \hat{p}(u)$ is the local Fisher information of the learned density, with eigenvalue bounds $\kappa$ (lower, the \emph{barrier curvature}) and $\Lambda$ (upper) formalized in Definition~\ref{def:barrier_curvature} and Assumption~\ref{ass:logconcave}.

\subsection{The Unknown Feasibility Manifold}
\label{sec:manifold}

Let $\mathcal{X}_{\mathrm{safe}} \subset \mathcal{X}$ be a closed set of safe states. We define the \textit{feasibility manifold} $\mathcal{M}_t(x_t)$ as the set of all inputs that maintain safety:
\begin{equation}
\label{eq:feasibility_manifold}
    \mathcal{M}_t(x_t) := \{ u \in \mathcal{U} \mid f_t(x_t, u) \in \mathcal{X}_{\mathrm{safe}} \}.
\end{equation}
Under smoothness of $f_t$, the set $\mathcal{M}_t$ constitutes a time-varying, potentially non-convex submanifold embedded in the control space $\mathcal{U}$, whose topology is governed by the nonlinear dynamics inside the Simulator. The manifold $\mathcal{M}_t$ is \textit{unknown to the Planner}, \textit{time-varying}, and accessible only through feasibility samples from the Simulator.

\begin{assumption}[Compactness and Regularity]
\label{ass:compactness}
    $\mathcal{X}_{\mathrm{safe}}$ is compact and $f_t$ is $C^2$-smooth for all $t$, so that $\mathcal{M}_t(x_t)$ is a compact embedded submanifold of $\mathcal{U}$.
\end{assumption}

Compactness of $\mathcal{X}_{\mathrm{safe}}$ is standard in safety-critical control~\cite{ames2016control}. The $C^2$-smoothness of $f_t$ ensures that the feasibility manifold $\mathcal{M}_t$ inherits a well-defined boundary, which is necessary for the score function to be well-defined in the interior.

\subsection{The Simulator--Planner Interface}
\label{sec:architecture}

The Simulator and Planner interact through a sample-based interface (\Cref{fig:architecture}). We formalize the black-box world model and the information flow between the two components.

\begin{definition}[Black-Box World Model]
\label{def:bbwm}
A \emph{black-box world model} is a feasibility oracle $\mathcal{W}_t : \mathcal{X} \times \mathcal{U} \to \{0, 1\}$ defined by $\mathcal{W}_t(x, u) = \bm{1}[f_t(x, u) \in \mathcal{X}_{\mathrm{safe}}]$. The dynamics $f_t$ are embedded inside the oracle and are not accessible in closed form. Given the current state $x_t$ and a proposal distribution $q(u)$, the oracle returns i.i.d.\ \emph{feasibility samples}\footnote{In practice, these samples can be obtained by any method that produces i.i.d.\ draws from the feasibility-conditioned distribution, e.g., rejection sampling from $q$ (used in our experiments; see~\Cref{sec:our_method}), MCMC, or direct forward simulation of the world model.} $\{u_i\}_{i=1}^N$ drawn from the conditional $q(u \mid \mathcal{W}_t(x_t, u) = 1)$. 
\end{definition}

The \textit{Simulator} has access to the full state $x_t$, the oracle $\mathcal{W}_t$, and the feasibility manifold $\mathcal{M}_t$. It produces feasibility samples and transmits a learned density estimate $\hat{p}_t$ (together with its score $\hat{s}_t = \nabla_u \ln \hat{p}_t$) to the Planner. The \textit{Planner} optimizes a stage cost $c(u)$ but has \textit{no access} to $f_t$, $x_t$, or the raw state; it receives only the score signal $\hat{s}_t$ and the cost function $c$. Formally, the information sets are
    \begin{align}
    \label{eq:info_sim}
    \mathcal{I}_t^{\mathrm{Sim}} &= \{ x_t,\, \mathcal{W}_t(\cdot),\, \mathcal{M}_t \}, \\
    \label{eq:info_pl}
    \mathcal{I}_t^{\mathrm{Pl}} &= \{ \hat{s}_t(\cdot),\, c(\cdot) \}.
    \end{align}
This information asymmetry ($f_t \notin \mathcal{I}_t^{\mathrm{Pl}}$) means the Planner cannot form constraint Jacobians or build explicit barrier functions; it must rely entirely on the Simulator's score signal.

\subsection{Score-Based Feasibility Signal}
\label{sec:score_learning}

Rather than transmitting the explicit manifold geometry to the Planner, the Simulator compresses $\mathcal{M}_t$ into a differentiable density whose score function encodes the manifold geometry. The ideal feasibility distribution and its score are from the existing information-aggregation literature~\cite{li2021information,li2021learning}.

\begin{definition}[Feasibility Distribution]
\label{def:feasibility_score}
The \emph{oracle feasibility distribution} $p^*_t(\cdot \mid x_t)$ is a $C^2$ density supported on the closure of the single-step feasibility manifold $\mathcal{M}_t(x_t)$, with $p^*_t > 0$ on $\mathrm{int}\,\mathcal{M}_t$, vanishing continuously on $\partial\mathcal{M}_t$, and having a nondegenerate gradient on $\partial\mathcal{M}_t$.\footnote{A canonical construction is a Gaussian-smoothed uniform density truncated to $\mathcal{M}_t$, i.e., $p^*_t \propto (\mathcal{N}(0,\sigma^2 I) * \bm{1}[\,\cdot\, \in \mathcal{M}_t])\cdot \bm{1}[\,\cdot\,\in\mathcal{M}_t]/Z$, which recovers the manifold projection and tangent-space operations needed for Riemannian optimization~\cite{kharitenko2025landing,hyvarinen2005score}. The smoothing scale $\sigma$ plays the role of the KDE bandwidth $h$ in Section~\ref{sec:kde_instantiation}.} The associated \emph{oracle score} is $s^*_t(u) := \nabla_u \ln p^*_t(u \mid x_t)$.
\end{definition}

We note that the maximum-entropy feasibility (MEF) density in~\cite{li2021information,li2021learning} assigns higher probability to actions that maximize \emph{future} flexibility by weighting with the volume of reachable feasible sets over a prediction horizon. Our single-step density is the special case in which no look-ahead is used. Extending the theory to multi-step MEF densities is an interesting direction but beyond the scope of this work.

The oracle score $s^*_t$ encodes the geometry of $\mathcal{M}_t$ in a form the Planner can use for gradient-based optimization, without ever accessing $f_t$. The critical challenge is that this score must be learned \textit{online} from streaming feasibility samples, rather than from a pretrained model.

\begin{assumption}[Bounded Information Loss]
\label{ass:bounded_loss}
    The Simulator employs a learning algorithm (e.g., KDE or score matching) to produce an estimate $\hat{p}_t$ of the feasibility density. The KL divergence between the oracle and the learned estimate is finite such that
$
    D_{KL}(p^*_t \| \hat{p}_t) < \infty.
$
\end{assumption}

This assumption is satisfied by any consistent estimator~\cite{tsybakov2009introduction}. We instantiate the estimator as a Gaussian KDE in Section~\ref{sec:kde_instantiation}, with convergence rates established in Section~\ref{sec:online_convergence}.
\begin{figure}[t]
    \centering
\begin{tikzpicture}[
    block/.style={draw, rounded corners, minimum height=0.9cm, minimum width=1.6cm, align=center, font=\scriptsize},
    arr/.style={-{Stealth[length=2mm]}, thick},
    lbl/.style={font=\tiny, midway},
    dashed_box/.style={draw, dashed, rounded corners, inner sep=4pt},
]
\node[block, fill=blue!8] (env) {Environment\\$x_{t+1}{=}f_t(x_t,u_t)$};
\node[block, fill=blue!8, right=0.4cm of env] (oracle) {World Model\\$\mathcal{W}_t(x_t,u)$};
\node[block, fill=blue!8, right=0.4cm of oracle] (kde) {KDE /\\Score Model};
\node[block, fill=orange!10, right=0.6cm of kde] (plan) {Planner\\$u_{\mathrm{ppc}}{\approx}\argmin \mathcal{F}$};

\draw[arr] (env) -- node[lbl,above]{$x_t$} (oracle);
\draw[arr] (oracle) -- node[lbl,above]{\{$u_i$\}} (kde);
\draw[arr] (kde) -- node[lbl,above]{$\hat{s}_t$} (plan);
\draw[arr, dashed] (plan.south) -- ++(0,-0.45) -| node[lbl, below, pos=0.25]{$u_{\mathrm{ppc}}$} (env.south);

\begin{scope}[on background layer]
\node[dashed_box, fit=(env)(oracle)(kde), label={[font=\tiny]above:Simulator}] {};
\end{scope}
\end{tikzpicture}
\caption{The Simulator--Planner architecture. The Simulator contains the environment dynamics and the feasibility oracle $\mathcal{W}_t$; it draws feasibility samples $\{u_i\}$ and fits a density model whose score $\hat{s}_t$ is transmitted to the Planner. The Planner minimizes the penalized objective $c(u) - \beta\ln\hat{p}_t(u)$ via gradient descent on $\hat{s}_t$, without access to $f_t$.}
    \label{fig:architecture}
\end{figure}
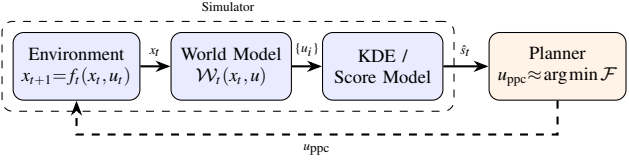

The Planner's goal is to minimize a stage cost $c(u_t)$ subject to the unknown manifold constraint $u_t \in \mathcal{M}_t$, using the \textit{online} score estimate $\hat{s}_t = \nabla_u \ln \hat{p}_t$ transmitted by the Simulator. The Planner does not observe the full state $x_t$ but instead receives a \emph{context signal} $\xi_t \in \mathbb{R}^d$, a fixed-dimensional summary of the environment. For example, in the robot navigation task of Section~\ref{sec:experiments}, $x_t$ includes all obstacle positions and velocities, while $\xi_t$ is a low-dimensional embedding of a top-down image produced by the Planner's sensor. The feasibility manifold $\mathcal{M}_t$ depends on $x_t$ through the obstacle geometry, but $\xi_t$ correlates with $x_t$ and allows the Planner to condition its density estimate as $\hat{p}(u \mid \xi_t)$, so that the learned manifold adapts when the environment changes. Without $\xi_t$, the Planner must rely on the marginal density $\hat{p}(u)$, which averages over all states and cannot track manifold shifts.

State-dependent objectives can be absorbed by conditioning the cost on $\xi_t$ (e.g., $c(u; \xi_t) = \|u - u_{\mathrm{goal}}(\xi_t)\|^2$ for some context-dependent target $u_{\mathrm{goal}}$); we suppress the $\xi_t$-dependence of $c$ for notational clarity. We formalize the fundamental problem as follows.

\textbf{Problem $\mathcal{P}$ (Contextual Control):}
Given a black-box world model $\mathcal{W}_t$ (Definition~\ref{def:bbwm}), a stage cost $c : \mathcal{U} \to \mathbb{R}$, and a context signal $\xi_t \in \mathbb{R}^d$, find a sequence of policies $\{\pi_t\}_{t=1}^T$ that solves
\begin{equation}\label{eq:main_problem}
    \min_{\{\pi_t\}} \;\; \sum_{t=1}^T \mathbb{E}_{u_t \sim \pi_t}\bigl[c(u_t)\bigr] \qquad \text{s.t.} \quad u_t \in \mathcal{M}_t, \;\; \forall\, t,
\end{equation}
using only feasibility samples from $\mathcal{W}_t$ and the context $\xi_t$, without access to the dynamics $f_t$.

Because the Planner observes $\xi_t$ rather than $x_t$, the oracle density from Definition~\ref{def:feasibility_score} induces two related distributions over actions. The \emph{conditional oracle density} $p^*(u \mid \xi_t) \coloneqq \mathbb{E}_{x_t \mid \xi_t}[p^*_t(u \mid x_t)]$ averages over the states consistent with the current context, while the \emph{marginal oracle density} $p^*(u) \coloneqq \int p^*(u \mid \xi)\, p(\xi)\, d\xi$ averages over all contexts under the context-generating distribution $p(\xi)$. Proposition~\ref{prop:ctx_curvature} and Theorem~\ref{thm:ctx_gap} quantify the benefit of working with the former rather than the latter, via the mixture-Hessian identity~\eqref{eq:mixture_hessian} and the resulting safety-guarantee gap~\eqref{eq:ctx_gap}.

Problem~$\mathcal{P}$ cannot be solved directly because the Planner cannot evaluate the hard constraint $u \in \mathcal{M}_t$, and the system dynamics couple the time steps through the state. Following~\cite{li2021information,li2021learning}, we decouple the problem by optimizing each $\pi_t$ independently, treating the current manifold $\mathcal{M}_t$ as given. At each step $t$, the Simulator compresses the constraint into the learned density $\hat{p}_t(\cdot \mid \xi_t)$, and safety is enforced by requiring the policy to remain close to this density. This yields the per-step relaxation:

\textbf{Problem $\mathcal{P}_{RD}$ (Per-Step Rate-Distortion Relaxation):}
\begin{equation}\label{eq:prd}
    \min_{\pi_t} \;\; \mathbb{E}_{u \sim \pi_t}\bigl[c(u)\bigr] \qquad \text{s.t.} \quad D_{KL}(\pi_t \| \hat{p}_t(\cdot \mid \xi_t)) \leq I_c.
\end{equation}
Solving $\mathcal{P}_{RD}$ at each $t$ and summing the costs recovers the structure of $\mathcal{P}$. When $\hat{p}_t \to p^*_t$ and $I_c \to 0$, the KL constraint forces $\pi_t$ to concentrate on $\mathcal{M}_t$, so per-step feasibility is recovered. The cumulative cost suboptimality of this online policy relative to the offline optimum of~$\mathcal{P}$ is bounded by Theorem~\ref{thm:dynamic_regret}. For notational brevity, we often write $\hat{p}_t(u)$ for $\hat{p}_t(u \mid \xi_t)$ when the context is clear; the full conditional notation is used whenever the distinction matters.

\section{Control Synthesis via Free Energy Minimization}
\label{sec:reformulation}

We now solve Problem~$\mathcal{P}_{RD}$. Since $\mathcal{P}_{RD}$ is strictly convex in $\pi$, the KKT conditions yield the penalty weight $\beta$ as the Lagrange multiplier of the capacity constraint $I_c$ (not an arbitrary tuning parameter). The critical \textit{stiffness condition} $\beta > \beta^*_{\mathrm{curv}}$ derived in~\Cref{sec:score_descent} ensures that the Planner adheres to the Simulator's learned manifold geometry with strictness proportional to the barrier curvature $\kappa$.

\begin{proposition}[Thermodynamic Reformulation]
\label{prop:thermo}
  The optimal policy $\pi^*_t(u)$ solving $\mathcal{P}_{RD}$ is the unique minimizer of the Free Energy functional:
  \begin{align}
    \mathcal{G}(\pi) = \mathbb{E}_{u \sim \pi} [ c(u) ] + \beta_t D_{KL}(\pi \| \hat{p}_t).  
  \end{align}
  The solution is given by the Gibbs-Boltzmann distribution:
  \begin{align*}
        \pi^*_t(u) = \frac{1}{Z_t} \hat{p}_t(u) \exp\left( -\frac{c(u)}{\beta_t} \right),
  \end{align*}
  where $\beta_t > 0$ is the Lagrange multiplier associated with the capacity constraint $I_c$, and
  \[
      Z_t = \int_{\mathcal{U}} \hat{p}_t(v)\exp\!\left(-\frac{c(v)}{\beta_t}\right)\,dv
  \]
  is the normalizing constant (partition function).
\end{proposition}

\Cref{prop:thermo} shows that the optimal policy balances economic performance against manifold adherence through a Gibbs-Boltzmann distribution. Here $\mathcal{G}(\pi)$ is a functional over policies; the PPC objective in the next subsection is the pointwise free-energy function $\mathcal{F}(u)=c(u)-\beta_t\ln\hat{p}_t(u)$ obtained by taking the MAP of $\pi^*_t$. The stiffness $\beta_t$ controls how strongly the Planner trusts the Simulator's feasibility signal.
\begin{proof}
We form the Lagrangian functional $\mathcal{L}(\pi, \beta, \lambda)$:
\begin{align}
\nonumber
    \mathcal{L} = & \int \pi(u) c(u) du + \beta \left( \int \pi(u) \ln \frac{\pi(u)}{\hat{p}(u)} du - I_c \right) \\
    \label{eq:lagrangian}
    &+ \lambda \left( \int \pi(u) du - 1 \right).
\end{align}
Taking the Fr\'{e}chet derivative with respect to $\pi(u)$ and setting to zero,
\[\frac{\delta \mathcal{L}}{\delta \pi} = c(u) + \beta (\ln \pi(u) - \ln \hat{p}(u) + 1) + \lambda = 0.\]
Rearranging the terms we obtain $\ln \pi(u) = \ln \hat{p}(u) - {c(u)}/{\beta} - ({\lambda}/{\beta} + 1)$. Exponentiating yields $\pi^*(u) \propto \hat{p}(u) e^{-c(u)/\beta}$, matching the optimal controllers in Linearly Solvable MDPs~\cite{todorov2006linearly} and path integral control~\cite{kappen2005linear}. 
\end{proof}

\subsection{Deterministic Control Law}

The Planner implements a single action $u_t$ by taking the maximum a posteriori (MAP) estimate of $\pi^*_t$, yielding the Penalized Predictive Control (PPC) objective~\cite{li2021information}. While PPC was originally derived for settings where the oracle density $p^*_t$ is available in closed form, our contribution is to instantiate it in the black-box regime where $p^*_t$ is replaced by the online estimate $\hat{p}_t$ learned from simulator samples, and to provide convergence and safety guarantees for this sample-based variant.

\begin{corollary}[The PPC Objective]
\label{cor:ppc}
    The MAP estimate $u_t^{\mathrm{MAP}} = \arg\max_{u} \pi^*_t(u)$ is equivalent to:
    \begin{equation}
    u_t^{\mathrm{MAP}} = \arg\min_{u \in \mathcal{U}} \left( c(u) - \beta_t \ln \hat{p}_t(u) \right).
    \end{equation}
\end{corollary}

\begin{proof}
By definition,
\begin{align*}
    u_t^{\mathrm{MAP}}
    &= \arg\max_{u}\, \frac{1}{Z_t}\, \hat{p}_t(u)\, e^{-c(u)/\beta_t} \\
    &= \arg\min_{u}\, \bigl( c(u)/\beta_t - \ln \hat{p}_t(u) \bigr).
\end{align*}
    Multiplying by $\beta_t$ yields the objective. Note that $u_t^{\mathrm{MAP}}$ is a conceptual unconstrained minimizer over $\mathcal{U}$; we reserve the notation $u_t^*$ for the restricted minimizer over $\hat{\mathcal{M}}_t^\alpha$ introduced in Lemma~\ref{lem:landscape}(c), which is the object that PPC gradient descent from a warm start in $\hat{\mathcal{M}}_t^\alpha$ converges to (Theorem~\ref{thm:contraction}). Proposition~\ref{prop:betacrit} further shows that $u_t^*$ is itself an unconstrained critical point of $\mathcal{F}$ (i.e., $\nabla\mathcal{F}(u_t^*)=0$) whenever $\beta_t > \beta^*_{\mathrm{curv}}$, although $\mathcal{F}$ may admit additional critical points outside $\hat{\mathcal{M}}_t^\alpha$.
\end{proof} 

The gradient of the PPC objective is
\begin{equation}
\nabla_u \left( c(u) - \beta_t \ln \hat{p}_t(u) \right) = \nabla_u c(u) - \beta_t \hat{s}_t(u),
\end{equation}
where $\hat{s}_t(u) = \nabla_u \ln \hat{p}_t(u)$ is the Simulator's learned score. Each gradient step thus combines the cost gradient $\nabla c$ (what is desirable) with the score $\hat{s}_t$ (what is feasible), and the stiffness $\beta_t$ mediates between minimizing $c$ and adhering to $\mathcal{M}_t$. The Planner never accesses $f_t$; all constraint information enters through $\hat{s}_t$.

The remainder of the paper answers Problem~$\mathcal{P}$ by analyzing the per-step free energy $\mathcal{F}_t(u) = c(u) - \beta_t \ln \hat{p}_t(u)$. When the time index is clear from context, we write $\mathcal{F}(u)$, $\hat{p}(u)$, and $\beta$ for brevity.

\section{Riemannian Score Descent}
\label{sec:score_descent}

In this section we fix a single time step $t$ and context $\xi_t$, analyzing the PPC free energy $\mathcal{F}(u) = c(u) - \beta \ln \hat{p}(u \mid \xi_t)$ with $\hat{p} \equiv \hat{p}_t(\cdot \mid \xi_t)$ and $\beta \equiv \beta_t$ held constant. Similarly, we simplify the learned $\alpha$-superlevel set $\hat{\mathcal{M}}_t^\alpha = \{u : \hat{p}_t(u \mid \xi_t) \geq \alpha\}$ as $\hat{\mathcal{M}}^\alpha$. The subscript $k$ indexes the inner gradient iterations within this step; the extension to time-varying densities is in Section~\ref{sec:dynamic_regret}.

\subsection{Density Estimator}
\label{sec:kde_instantiation}
We instantiate the Simulator's estimator $\hat{p}_t$ (Assumption~\ref{ass:bounded_loss}) as a Gaussian kernel density estimator (KDE)~\cite{tsybakov2009introduction}, which admits a closed-form score and supports the convergence analysis of Section~\ref{sec:online_convergence}. At time $t$, the Simulator has collected $N_t$ cumulative feasibility samples $\{U_1, \ldots, U_{N_t}\} \subset \mathbb{R}^m$, and estimates the marginal feasibility density by
\begin{equation}\label{eq:kde}
    \hat{p}_{N}(u) = \frac{1}{N}\sum_{i=1}^{N} K_h(u - U_i), \quad K_h(z) = \frac{e^{-\|z\|^2/(2h^2)}}{(2\pi h^2)^{m/2}},
\end{equation}
with bandwidth $h = \Theta(N^{-1/(m+4)})$. Setting $\hat{p}_t = \hat{p}_{N_t}$ gives the score $\hat{s}_t(u) = \nabla_u \ln \hat{p}_t(u)$ in closed form. When the Planner observes context $\xi_t \in \mathbb{R}^d$, the Simulator stores joint pairs $(U_i, \xi_i) \in \mathbb{R}^m \times \mathbb{R}^d$ and fits a product-kernel KDE,
\[
    \hat{p}(u, \xi) = \frac{1}{N}\sum_{i=1}^N K_{h_u}(u - U_i)\, K_{h_\xi}(\xi - \xi_i),
\]
where $K_{h_u}$ and $K_{h_\xi}$ are Gaussian kernels with bandwidths $h_u$ and $h_\xi$ for the action and context dimensions respectively. The conditional density $\hat{p}(u \mid \xi_t) = \hat{p}(u, \xi_t) / \int \hat{p}(v, \xi_t)\, dv$ and its score $\nabla_u \ln \hat{p}(u \mid \xi_t)$ are used in place of $\hat{p}_t(u)$ and $\hat{s}_t(u)$ throughout. For the KDE~\eqref{eq:kde}, the KL divergence in Assumption~\ref{ass:bounded_loss} decreases with $N_t$ and is bounded by the integrated squared score error (Theorem~\ref{thm:learning}).

\subsection{The Riemannian Structure of PPC}
This subsection establishes the geometric structure that drives PPC convergence and safety. In particular, it links local density curvature to optimization curvature through the PPC Hessian.
Recall from Corollary~\ref{cor:ppc} that PPC minimizes the free energy $\mathcal{F}(u) = c(u) - \beta \ln \hat{p}(u)$. The \emph{PPC gradient iterates} are defined by
\begin{equation}\label{eq:ppc_iterates}
    u_{k+1} = u_k - \eta\,\nabla\mathcal{F}(u_k), \quad k = 0, 1, \ldots, K{-}1,
\end{equation}
with step size $\eta > 0$, inner step count $K$, and warm start $u_0 = u_{t-1}$. The gradient is $\nabla \mathcal{F} = \nabla c - \beta \hat{s}$, and the Hessian is
\begin{equation}
\label{eq:ppc_hessian}
    H_{\mathcal{F}}(u) = \nabla^2 c(u) + \beta\, \mathcal{I}_{\hat{p}}(u),
\end{equation}
where $\mathcal{I}_{\hat{p}}(u) := -\nabla^2 \ln \hat{p}(u)$ is the \emph{local Fisher information} of the learned density. Thus each PPC gradient step implicitly uses a Riemannian metric shaped by the curvature of the log-density barrier. The action space $\mathcal{U}$, equipped with the positive-definite matrix field $\mathcal{I}_{\hat{p}}(u)$, forms a Riemannian manifold whose metric is the Fisher-Rao metric of the learned density $\hat{p}(\cdot \mid \xi)$. Gradient descent on $\mathcal{F}$ is then equivalent to Euclidean gradient descent with a geometry-aware preconditioner. Near the manifold boundary, $\mathcal{I}_{\hat{p}}$ is large and the effective step size shrinks, preventing the iterate from leaving the feasible region; in the interior, $\mathcal{I}_{\hat{p}}$ is small and the step expands. This geometry-awareness is a structural advantage over CBF-QP~\cite{ames2016control}, which uses a flat Euclidean metric in its safety projection.

Given the learned density $\hat{p}$ and a threshold $\alpha > 0$, let $\hat{\mathcal{M}}^\alpha$ denote the learned $\alpha$-superlevel set. The threshold $\alpha$ controls the conservatism of the learned safe set and is chosen by the Simulator (e.g., as a percentile of the KDE density values; see Algorithm~\ref{alg:nfc}). We now formalize the central geometric quantity of the paper.

\begin{definition}[Barrier Curvature]
\label{def:barrier_curvature}
Let $\hat{p}$ be a learned feasibility density and $\hat{\mathcal{M}}^\alpha = \{u : \hat{p}(u) \geq \alpha\}$ its $\alpha$-superlevel set. The \emph{barrier curvature} $\kappa$ and the \emph{curvature upper bound} $\Lambda$ are the tightest constants satisfying
\begin{equation}\label{eq:curvature_bounds}
    \kappa\, I \;\preceq\; \mathcal{I}_{\hat{p}}(u) \;\preceq\; \Lambda\, I
    \qquad \forall\, u \in \hat{\mathcal{M}}^\alpha,
\end{equation}
where $\mathcal{I}_{\hat{p}}(u) = -\nabla^2 \ln \hat{p}(u)$ is the local Fisher information defined in~\eqref{eq:ppc_hessian}. When the density is conditioned on context $\xi_t$, the curvature bounds depend on $\xi_t$ and we write
\[
    \kappa(\xi_t) \;:=\; \min_{u \in \hat{\mathcal{M}}_t^\alpha}\, \lambda_{\min}\!\bigl(\mathcal{I}_{\hat{p}(\cdot|\xi_t)}(u)\bigr).
\]
When the context is suppressed (as in the single-step analysis of this section), $\kappa$ refers to the barrier curvature of whichever density is in use.
\end{definition}

The barrier curvature $\kappa$ quantifies how sharply the log-density increases toward the interior of $\hat{\mathcal{M}}^\alpha$; it is the analogue, in our score-based setting, of the barrier self-concordance parameter in interior-point methods. Its upper counterpart $\Lambda$ controls the smoothness of the free energy $\mathcal{F}$.

\begin{assumption}[Local Log-Concavity]
\label{ass:logconcave}
    $\kappa > 0$ (equivalently, $-\ln\hat{p}$ is strictly convex on $\hat{\mathcal{M}}^\alpha$), $\Lambda < \infty$, and $\hat{\mathcal{M}}^\alpha$ is a single connected component of the superlevel set $\{u : \hat{p}(u) \geq \alpha\}$ that is convex (equivalently, $\alpha$ exceeds the density value at every saddle point of $\hat{p}$ so that the component lies in a single basin). The cost $c$ is convex, $L_c$-smooth, and has bounded gradient $\|\nabla c\|_\infty \leq G_c$.
\end{assumption}

The condition $\kappa > 0$ requires $-\ln\hat{p}$ to be strictly convex on $\hat{\mathcal{M}}^\alpha$. For the Gaussian KDE~\eqref{eq:kde} used throughout this paper, $\hat{p}(u) = \frac{1}{N}\sum_{i=1}^N K_h(u - U_i)$ is a mixture of Gaussians, and $-\ln\hat{p}$ is smooth everywhere. The single-basin/convexity clause is not automatic for mixtures: a Gaussian mixture can have several modes, and if $\alpha$ is taken below the density at an inter-mode saddle, $\{\hat{p} \geq \alpha\}$ becomes a disjoint union of convex components with the Hessian indefinite between them. Choosing $\alpha$ above every saddle value (equivalently, setting $\alpha$ as a high percentile of observed density values) isolates one basin; within that basin $\hat{p}$ is log-concave, $-\ln\hat{p}$ is strictly convex, and the superlevel set is convex. The Hessian $\mathcal{I}_{\hat{p}}(u)$ is then positive definite throughout, and $\kappa > 0$ holds with $\kappa$ determined by the narrowest bottleneck of the set. For the Gaussian KDE, $\Lambda = 1/h^2$ where $h$ is the bandwidth, so $\Lambda < \infty$ holds automatically. Both $\kappa$ and $\Lambda$ are computable from the KDE in closed form, and the single-basin condition can be checked online by locating the critical points of $\hat{p}$, so Assumption~\ref{ass:logconcave} is \emph{verifiable online} by the Simulator at each step.

The following lemma establishes the landscape properties of $\mathcal{F}$ that underpin the convergence theory.

\begin{lemma}[Free Energy Landscape]
\label{lem:landscape}
Under Assumption~\ref{ass:logconcave}, let $\mu := \beta\kappa$ and $L := L_c + \beta\Lambda$. The PPC free energy $\mathcal{F}(u) = c(u) - \beta\ln\hat{p}(u)$ satisfies:
\begin{enumerate}
    \item[(a)] $\mu$-strong convexity: $H_{\mathcal{F}}(u) \succeq \mu\, I$ for all $u \in \hat{\mathcal{M}}^\alpha$;
    \item[(b)] $L$-smoothness: $H_{\mathcal{F}}(u) \preceq L\, I$ for all $u \in \hat{\mathcal{M}}^\alpha$;
    \item[(c)] \emph{local} uniqueness: $\mathcal{F}$ has a unique minimizer on the compact convex set $\hat{\mathcal{M}}^\alpha$, denoted $u^* := \arg\min_{u\in\hat{\mathcal{M}}^\alpha}\mathcal{F}(u)$, and at most one unconstrained critical point of $\mathcal{F}$ lies in $\hat{\mathcal{M}}^\alpha$; when such a point exists it equals $u^*$ (equivalently, $u^* \in \mathrm{int}(\hat{\mathcal{M}}^\alpha)$ and $\nabla\mathcal{F}(u^*)=0$);
    \item[(d)] Polyak--\L{}ojasiewicz inequality:
    $\|\nabla\mathcal{F}(u)\|^2 \geq 2\mu\bigl(\mathcal{F}(u) - \mathcal{F}(u^*)\bigr)$ for all $u \in \hat{\mathcal{M}}^\alpha$.
\end{enumerate}
The condition number $\chi := L/\mu = (L_c + \beta\Lambda)/(\beta\kappa)$ governs the convergence rate of gradient-based iterates on $\mathcal{F}$. Note that (c) asserts uniqueness only within $\hat{\mathcal{M}}^\alpha$; since $\hat{p}$ is a Gaussian mixture, $\mathcal{F}$ need not be convex outside $\hat{\mathcal{M}}^\alpha$ and may admit additional critical points there.
\end{lemma}

\begin{proof}
For every $u \in \hat{\mathcal{M}}^\alpha$, the Hessian decomposes as $H_{\mathcal{F}}(u) = \nabla^2 c(u) + \beta\,\mathcal{I}_{\hat{p}}(u)$ by~\eqref{eq:ppc_hessian}.

\emph{Part~(a).} Since $c$ is convex, $\nabla^2 c(u) \succeq 0$. By Assumption~\ref{ass:logconcave}, $\mathcal{I}_{\hat{p}}(u) \succeq \kappa\, I$. Hence
\begin{equation}\label{eq:strong_cvx}
    H_{\mathcal{F}}(u) \;\succeq\; 0 + \beta\kappa\, I \;=\; \mu\, I.
\end{equation}

\emph{Part~(b).} By $L_c$-smoothness of $c$, $\nabla^2 c(u) \preceq L_c\, I$. By Assumption~\ref{ass:logconcave}, $\mathcal{I}_{\hat{p}}(u) \preceq \Lambda\, I$. Hence
$
    H_{\mathcal{F}}(u) \;\preceq\; L_c\, I + \beta\Lambda\, I \;=\; L\, I.
$

\emph{Part~(c).} $\hat{\mathcal{M}}^\alpha$ is closed (as a superlevel set of the continuous function $\hat{p}$) and bounded (since $-\ln\hat{p}$ is coercive for the Gaussian KDE, the superlevel set $\{\hat{p} \geq \alpha\}$ is contained in a ball of finite radius), hence compact, and convex by Assumption~\ref{ass:logconcave}. $\mathcal{F}$ is $\mu$-strongly convex on this set by part~(a); a continuous strongly convex function on a compact convex set attains its minimum at exactly one point, which we denote $u^*$. Strong convexity further implies that any $u' \in \hat{\mathcal{M}}^\alpha$ with $\nabla\mathcal{F}(u') = 0$ satisfies $\mathcal{F}(u) \geq \mathcal{F}(u') + \tfrac{\mu}{2}\|u-u'\|^2$ for all $u \in \hat{\mathcal{M}}^\alpha$, forcing $u' = u^*$.

\emph{Part~(d).} Fix $u \in \hat{\mathcal{M}}^\alpha$. By Taylor expansion and~\eqref{eq:strong_cvx}, for any $v \in \hat{\mathcal{M}}^\alpha$:
$\mathcal{F}(v) \geq \mathcal{F}(u) + \langle\nabla\mathcal{F}(u), v{-}u\rangle + \frac{\mu}{2}\|v{-}u\|^2$.
Minimizing the right-hand side over $v$ (unconstrained quadratic) yields the lower bound $\mathcal{F}(u^*) \geq \mathcal{F}(u) - \frac{1}{2\mu}\|\nabla\mathcal{F}(u)\|^2$. Rearranging gives the Polyak--\L{}ojasiewicz inequality.
\end{proof}

The landscape established by Lemma~\ref{lem:landscape} directly yields a geometric convergence guarantee. The strong convexity provided by the barrier curvature $\kappa$ forces PPC iterates toward the unique minimizer at a rate controlled by the condition number $\chi$.

\begin{theorem}[Score-Descent Contraction]
\label{thm:contraction} Set $\mu := \beta\kappa$, $L := L_c + \beta\Lambda$, $\chi := L/\mu$.
Let $u^* := \arg\min_{u\in\hat{\mathcal{M}}^\alpha}\mathcal{F}(u)$ be the restricted minimizer of Lemma~\ref{lem:landscape}(c). Suppose Assumption~\ref{ass:logconcave} holds, $u^*\in\mathrm{int}(\hat{\mathcal{M}}^\alpha)$, and the warm start satisfies
\[
    u_0\in\hat{\mathcal{M}}^\alpha, \qquad \|u_0 - u^*\|\leq\mathrm{dist}(u^*,\partial\hat{\mathcal{M}}^\alpha).
\]
Then the PPC iterates $\{u_k\}$ of~\eqref{eq:ppc_iterates} with step size $\eta\leq 1/L$ satisfy $u_k\in\hat{\mathcal{M}}^\alpha$ for all $k$, together with the following contractions for every $K\geq 0$:
\begin{enumerate}
\item[(i)] \emph{Function-value contraction.}
\begin{equation}\label{eq:func_contraction}
    \mathcal{F}(u_K) - \mathcal{F}(u^*) \;\leq\; (1-\eta\mu)^K \bigl[\mathcal{F}(u_0)-\mathcal{F}(u^*)\bigr].
\end{equation}
\item[(ii)] \emph{Iterate contraction.}
\begin{equation}\label{eq:contraction}
    \|u_K - u^*\| \;\leq\; \sqrt{\chi}\,(1-\eta\mu)^{K/2}\,\|u_0 - u^*\|.
\end{equation}
\end{enumerate}
\end{theorem}

\begin{proof}
\emph{Iterates remain in $\hat{\mathcal{M}}^\alpha$.} We first show that $u_k \in \hat{\mathcal{M}}^\alpha$ for every $k$ so that Lemma~\ref{lem:landscape} applies at each iterate. By $\mu$-strong convexity and $L$-smoothness of $\mathcal{F}$ on $\hat{\mathcal{M}}^\alpha$, for $\eta \leq 1/L$ the standard non-expansive bound $\|u_{k+1} - u^*\|^2 \leq (1-\eta\mu)^2\|u_k - u^*\|^2$ holds whenever $u_k \in \hat{\mathcal{M}}^\alpha$~\cite{boumal2023introduction} (by the contractivity of $I - \eta\nabla^2\mathcal{F}$ whose eigenvalues lie in $[1-\eta L, 1-\eta\mu]\subset[0,1-\eta\mu]$), i.e., iterate distance to $u^*$ is monotonically non-increasing. Since $u_0 \in \hat{\mathcal{M}}^\alpha$ and $\|u_0 - u^*\| \leq \mathrm{dist}(u^*, \partial\hat{\mathcal{M}}^\alpha)$, every $u_k$ lies in the closed Euclidean ball $B(u^*, \|u_0 - u^*\|) \subseteq \hat{\mathcal{M}}^\alpha$, using that $\hat{\mathcal{M}}^\alpha$ is convex (Assumption~\ref{ass:logconcave}) and contains this ball by construction. A simple induction then confirms $u_k \in \hat{\mathcal{M}}^\alpha$ for all $k \geq 0$.

By the $L$-smoothness of $\mathcal{F}$ (Lemma~\ref{lem:landscape}(b)),
\begin{equation}
\nonumber
    \mathcal{F}(u_{k+1}) \;\leq\; \mathcal{F}(u_k) + \langle\nabla\mathcal{F}(u_k),\, u_{k+1}{-}u_k\rangle + \frac{L}{2}\|u_{k+1}{-}u_k\|^2.
\end{equation}
Substituting the gradient step $u_{k+1} = u_k - \eta\,\nabla\mathcal{F}(u_k)$,
\begin{align}
\nonumber
    \mathcal{F}(u_{k+1}) &\leq \mathcal{F}(u_k) - \eta\|\nabla\mathcal{F}(u_k)\|^2 + \frac{\eta^2 L}{2}\|\nabla\mathcal{F}(u_k)\|^2 \\
\label{eq:descent2}
    &= \mathcal{F}(u_k) - \eta\Bigl(1 - \frac{\eta L}{2}\Bigr)\|\nabla\mathcal{F}(u_k)\|^2.
\end{align}
For $\eta \leq 1/L$, the factor $1 - \eta L/2 \geq 1/2 > 0$, so every gradient step strictly decreases $\mathcal{F}$.

The Polyak--\L{}ojasiewicz inequality from Lemma~\ref{lem:landscape}(d) gives
\begin{equation}\label{eq:pl}
    \|\nabla\mathcal{F}(u_k)\|^2 \;\geq\; 2\mu\bigl(\mathcal{F}(u_k) - \mathcal{F}(u^*)\bigr).
\end{equation}

Combining these, substituting~\eqref{eq:pl} into~\eqref{eq:descent2} with $1 - \eta L/2 \geq 1/2$,
\begin{align}
\nonumber
    \mathcal{F}(u_{k+1}) - \mathcal{F}(u^*) &\leq \bigl(\mathcal{F}(u_k) - \mathcal{F}(u^*)\bigr) - \eta\mu\bigl(\mathcal{F}(u_k) - \mathcal{F}(u^*)\bigr) \\
\label{eq:func_rate}
    &= (1 - \eta\mu)\bigl(\mathcal{F}(u_k) - \mathcal{F}(u^*)\bigr).
\end{align}
Iterating~\eqref{eq:func_rate} over $K$ steps yields part~(i).

For the iterate distance, by $\mu$-strong convexity and $L$-smoothness,
\begin{equation}\label{eq:sandwich}
    \tfrac{\mu}{2}\|u-u^*\|^2 \;\leq\; \mathcal{F}(u) - \mathcal{F}(u^*) \;\leq\; \tfrac{L}{2}\|u-u^*\|^2.
\end{equation}
Combining~\eqref{eq:sandwich} with part~(i) gives
\[
    \tfrac{\mu}{2}\|u_K-u^*\|^2 \;\leq\; (1-\eta\mu)^K \cdot \tfrac{L}{2}\|u_0-u^*\|^2.
\]
Dividing by $\mu/2$ and taking square roots yields part~(ii) with $\chi = L/\mu$.
\end{proof}

Theorem~\ref{thm:contraction} ensures fast convergence of PPC iterates (warm-started within $\hat{\mathcal{M}}^\alpha$) to the restricted minimizer $u^*$ of Lemma~\ref{lem:landscape}(c), but takes the interior condition $u^* \in \mathrm{int}(\hat{\mathcal{M}}^\alpha)$ as a hypothesis. Whether this condition actually holds is a separate question, since $u^*$ may sit safely interior to $\hat{\mathcal{M}}^\alpha$ or merely on its boundary. Because the KDE $\hat{p}$ is a sum of Gaussians, $\mathcal{F}$ is finite on all of $\mathcal{U}$ and may possess additional critical points outside $\hat{\mathcal{M}}^\alpha$; when $\beta$ is small, the cost gradient dominates the score barrier at the boundary, so the minimum of $\mathcal{F}$ over $\hat{\mathcal{M}}^\alpha$ is pinned to $\partial\hat{\mathcal{M}}^\alpha$, and unconstrained gradient descent on $\mathcal{F}$ drifts out of the high-density region. Proposition~\ref{prop:betacrit} below supplies a sufficient condition on the stiffness, $\beta > \beta^*_{\mathrm{curv}}$, under which $u^*$ is strictly \emph{interior} to $\hat{\mathcal{M}}^\alpha$, or equivalently an unconstrained critical point of $\mathcal{F}$, and hence by Lemma~\ref{lem:landscape}(c) the unique such point within $\hat{\mathcal{M}}^\alpha$. This closes the interior hypothesis of Theorem~\ref{thm:contraction}. Proposition~\ref{prop:betacrit} relies only on Lemma~\ref{lem:landscape} and Assumption~\ref{ass:logconcave}, not on Theorem~\ref{thm:contraction}, so there is no circularity.

Under Assumption~\ref{ass:logconcave}, the negative log-density $g(u)=-\ln\hat{p}(u)$ is $\kappa$-strongly convex on $\hat{\mathcal{M}}^\alpha$, so it has a unique minimizer $\bar{u} := \arg\max_{u\in\hat{\mathcal{M}}^\alpha}\hat{p}(u)$, the point of highest learned density. Define the \emph{level-set radius} $r_\alpha = \inf_{u \in \partial \hat{\mathcal{M}}^\alpha}\|u - \bar{u}\| > 0$.

\begin{proposition}[Critical Stiffness from Manifold Curvature]
\label{prop:betacrit}
Under Assumption~\ref{ass:logconcave}, let $u^* := \arg\min_{u\in\hat{\mathcal{M}}^\alpha}\mathcal{F}(u)$ denote the restricted minimizer from Lemma~\ref{lem:landscape}(c). Whenever the stiffness satisfies
\begin{equation}\label{eq:beta_crit}
    \beta \;>\; \beta^*_{\mathrm{curv}} \;:=\; \frac{G_c}{\kappa \cdot r_\alpha},
\end{equation}
the following hold (the threshold $\beta = \beta^*_{\mathrm{curv}}$ is a limiting case in which the interior margin degenerates to zero):
\begin{enumerate}
    \item[(i)] \emph{Interior placement.} $u^* \in \mathrm{int}(\hat{\mathcal{M}}^\alpha)$, with displacement and margin bounds
    \[
        \|u^*-\bar{u}\| \;\leq\; \frac{G_c}{\beta\kappa},
        \qquad
        \mathrm{dist}(u^*, \partial\hat{\mathcal{M}}^\alpha) \;\geq\; r_\alpha - \frac{G_c}{\beta\kappa} \;>\; 0.
    \]
    \item[(ii)] \emph{Unconstrained stationarity.} $\nabla\mathcal{F}(u^*) = 0$, so $u^*$ is also an unconstrained critical point of $\mathcal{F}$, and it is the unique critical point of $\mathcal{F}$ inside $\hat{\mathcal{M}}^\alpha$.
\end{enumerate}
\end{proposition}

\begin{proof}
Let $g(u) := -\ln \hat{p}(u)$. By Assumption~\ref{ass:logconcave}, $g$ is $\kappa$-strongly convex on $\hat{\mathcal{M}}^\alpha$ with $\bar{u}$ its unique minimizer on the set. Since $r_\alpha > 0$ by its definition, $\bar{u} \in \mathrm{int}(\hat{\mathcal{M}}^\alpha)$, so the unconstrained stationarity $\nabla g(\bar{u}) = 0$ holds. Strong convexity gives the monotonicity inequality $\langle \nabla g(u) - \nabla g(\bar{u}),\, u-\bar{u}\rangle \geq \kappa\|u-\bar{u}\|^2$; combining with $\nabla g(\bar{u}) = 0$ and applying Cauchy--Schwarz to the left-hand side yields
\begin{equation}\label{eq:score_lower}
    \|\nabla g(u)\| \;\geq\; \kappa\|u-\bar{u}\|
    \quad \forall\, u \in \hat{\mathcal{M}}^\alpha,
\end{equation}
and hence $\|\hat{s}(u)\| = \|\nabla g(u)\| \geq \kappa\|u-\bar{u}\|$ on $\hat{\mathcal{M}}^\alpha$.

\emph{Part (i): interior placement.} By Lemma~\ref{lem:landscape}(c), $u^* = \arg\min_{u\in\hat{\mathcal{M}}^\alpha}\mathcal{F}(u)$ exists and is unique. We rule out the boundary case by contradiction.

Suppose $u^* \in \partial\hat{\mathcal{M}}^\alpha$, so $\|u^*-\bar{u}\| \geq r_\alpha$. Since $\hat{\mathcal{M}}^\alpha = \{u : g(u) \leq -\ln\alpha\}$ and $g$ is smooth and convex, the outward unit normal to $\hat{\mathcal{M}}^\alpha$ at $u^*$ is $n = \nabla g(u^*)/\|\nabla g(u^*)\|$. First-order optimality of the constrained minimizer $u^*$ requires that no feasible direction improves $\mathcal{F}$; equivalently, $\langle \nabla\mathcal{F}(u^*), n\rangle \leq 0$. Expanding $\nabla\mathcal{F} = \nabla c + \beta\nabla g$,
\[
    \langle \nabla c(u^*),\, n\rangle + \beta \|\nabla g(u^*)\| \;\leq\; 0.
\]
Using $\langle \nabla c(u^*), n\rangle \geq -\|\nabla c(u^*)\| \geq -G_c$ and $\|\nabla g(u^*)\| \geq \kappa\|u^*-\bar{u}\| \geq \kappa r_\alpha$ from~\eqref{eq:score_lower}, this forces $\beta \kappa r_\alpha \leq G_c$, i.e., $\beta \leq \beta^*_{\mathrm{curv}}$, contradicting the hypothesis $\beta > \beta^*_{\mathrm{curv}}$.

Therefore $u^* \in \mathrm{int}(\hat{\mathcal{M}}^\alpha)$ for $\beta > \beta^*_{\mathrm{curv}}$, and unconstrained stationarity $\nabla\mathcal{F}(u^*) = 0$ gives $\nabla c(u^*) = \beta\,\hat{s}(u^*)$. Taking norms and using~\eqref{eq:score_lower},
\begin{equation}
    \|u^*-\bar{u}\| \;\leq\; \frac{\|\hat{s}(u^*)\|}{\kappa}
    \;=\; \frac{\|\nabla c(u^*)\|}{\beta\kappa}
    \;\leq\; \frac{G_c}{\beta\kappa}.
\end{equation}
Because $\bar{u}\in\mathrm{int}(\hat{\mathcal{M}}^\alpha)$ with nearest boundary point at distance $r_\alpha$, the triangle inequality gives $\mathrm{dist}(u^*, \partial\hat{\mathcal{M}}^\alpha) \geq r_\alpha - G_c/(\beta\kappa) > 0$ for $\beta > \beta^*_{\mathrm{curv}}$.

\emph{Part (ii): unconstrained stationarity and uniqueness inside $\hat{\mathcal{M}}^\alpha$.} Interior placement of $u^*$ makes the constraint inactive, so $\nabla\mathcal{F}(u^*) = 0$ holds in the unconstrained sense. Any other point $v \in \hat{\mathcal{M}}^\alpha$ with $\nabla\mathcal{F}(v) = 0$ would, by Lemma~\ref{lem:landscape}(c), satisfy $v = u^*$. Thus $u^*$ is the unique critical point of $\mathcal{F}$ within $\hat{\mathcal{M}}^\alpha$, although $\mathcal{F}$ may possess additional critical points in $\mathcal{U}\setminus\hat{\mathcal{M}}^\alpha$ where log-concavity need not hold.
\end{proof}

The key mechanism is that strong convexity anchors $u^*$ near $\bar{u}$, which is known to be interior to $\hat{\mathcal{M}}^\alpha$ with margin $r_\alpha$. Above the threshold $\beta^*_{\mathrm{curv}}$, the cost-driven displacement $G_c/(\beta\kappa)$ stays below $r_\alpha$, so $u^*$ cannot reach $\partial\hat{\mathcal{M}}^\alpha$; equivalently, the inward pull of the score barrier dominates the outward pull of the cost gradient everywhere on $\partial\hat{\mathcal{M}}^\alpha$, ruling out boundary optima. The critical stiffness $\beta^*_{\mathrm{curv}}$ depends on the barrier curvature $\kappa$ and the level-set radius $r_\alpha$, both of which are properties of the learned density, rather than on a Lipschitz bound on the unknown dynamics. Both $\kappa$ and $r_\alpha$ can be estimated online from the KDE (see~\Cref{sec:our_method}).

\subsection{Dynamic Regret under Manifold Drift}
\label{sec:dynamic_regret}

When the manifold is time-varying, $\mathcal{F}_t$ changes at every step. We use the \emph{dynamic regret} as our performance metric, which is defined as the cumulative suboptimality relative to the moving restricted minimizer $u^*_t := \arg\min_{u\in\hat{\mathcal{M}}_t^\alpha}\mathcal{F}_t(u)$. Under $\beta > \beta^*_{\mathrm{curv}}$, Proposition~\ref{prop:betacrit} further identifies $u^*_t$ with the unique unconstrained critical point of $\mathcal{F}_t$ in $\hat{\mathcal{M}}_t^\alpha$, so stationarity $\nabla\mathcal{F}_t(u^*_t) = 0$ is available for the analysis. We first bound the sensitivity of $u^*_t$ to score changes, then derive the regret bound.

\begin{lemma}[Comparator Sensitivity]
\label{lem:comparator}
Under the conditions of Lemma~\ref{lem:landscape} and $\beta > \beta^*_{\mathrm{curv}}$ (so each $u^*_t \in \mathrm{int}(\hat{\mathcal{M}}_t^\alpha)$ by Proposition~\ref{prop:betacrit}), suppose $u^*_{t-1}\in\hat{\mathcal{M}}_t^\alpha$ for every $t$. Then the comparator path length $\mathcal{C}_T := \sum_{t=1}^T \|u^*_t - u^*_{t-1}\|$ satisfies
\begin{equation}\label{eq:comparator}
    \mathcal{C}_T \;\leq\; \frac{1}{\kappa}\sum_{t=1}^T \sup_u \|\hat{s}_t(u) - \hat{s}_{t-1}(u)\|.
\end{equation}
\end{lemma}

\begin{proof}
Interior stationarity $\nabla\mathcal{F}_t(u^*_t)=0$ holds at each step by Proposition~\ref{prop:betacrit}. From $\nabla\mathcal{F}_{t-1}(u^*_{t-1})=0$, i.e., $\nabla c(u^*_{t-1}) = \beta\,\hat{s}_{t-1}(u^*_{t-1})$, evaluate the gradient of $\mathcal{F}_t$ at the old minimizer:
\begin{align}
\nonumber
    \|\nabla\mathcal{F}_t(u^*_{t-1})\| &= \|\nabla c(u^*_{t-1}) - \beta\,\hat{s}_t(u^*_{t-1})\| \\
    \nonumber
    &= \beta\|\hat{s}_{t-1}(u^*_{t-1}) - \hat{s}_t(u^*_{t-1})\| \\
\label{eq:grad_shift}
    &\leq \beta\sup_u\|\hat{s}_t(u) - \hat{s}_{t-1}(u)\|.
\end{align}
By $\mu$-strong convexity of $\mathcal{F}_t$ (Lemma~\ref{lem:landscape}(a)) together with $\nabla\mathcal{F}_t(u^*_t)=0$ (interior stationarity at step $t$),
\begin{align}\label{eq:sc_bound}
    \|\nabla\mathcal{F}_t(u^*_{t-1})\|
    &= \|\nabla\mathcal{F}_t(u^*_{t-1})-\nabla\mathcal{F}_t(u^*_t)\| \nonumber \\
    &\geq \mu\,\|u^*_{t-1} - u^*_t\|
    \;=\; \beta\kappa\,\|u^*_{t-1} - u^*_t\|.
\end{align}
Combining \eqref{eq:grad_shift} and \eqref{eq:sc_bound} yields $\|u^*_t - u^*_{t-1}\| \leq \sup_u\|\hat{s}_t - \hat{s}_{t-1}\|/\kappa$. Summing over $t$ gives~\eqref{eq:comparator}.
\end{proof}

The hypothesis $u^*_{t-1}\in\hat{\mathcal{M}}_t^\alpha$ is mild. By Proposition~\ref{prop:betacrit}, $u^*_{t-1}$ sits interior to $\hat{\mathcal{M}}_{t-1}^\alpha$ with margin $r_\alpha - G_c/(\beta\kappa) > 0$, so it remains in $\hat{\mathcal{M}}_t^\alpha$ whenever the feasible set drifts by less than that margin, the natural regime in which dynamic regret is informative.

Theorem~\ref{thm:contraction} treats a fixed manifold. In practice, the feasibility set drifts over time. The following result quantifies how the cumulative suboptimality depends on both the contraction rate and the speed of manifold drift, measured through the \emph{score variation} $\mathcal{V}_T := \sum_{t=1}^T \sup_u \|\hat{s}_t(u) - \hat{s}_{t-1}(u)\|^2$.

\begin{theorem}[Dynamic Regret under Manifold Drift]
\label{thm:dynamic_regret}
Let $\{u_t\}$ be the output of PPC with $K$ inner gradient steps (step size $\eta=1/L$) warm-started from $u_{t-1}$, and let
$u_t^* := \arg\min_{u\in\hat{\mathcal{M}}_t^\alpha}\mathcal{F}_t(u)$
be the restricted minimizer of $\mathcal{F}_t(u) := c(u) - \beta\ln\hat{p}_t(u)$.
Assume, for every $t$:
\begin{enumerate}
    \item[(i)] Assumptions~\ref{ass:compactness}--\ref{ass:logconcave} hold for $\hat{p}_t$;
    \item[(ii)] $\beta > \beta^*_{\mathrm{curv}}$;
    \item[(iii)] (Warm-start margin) Let $R^* := r_\alpha - G_c/(\beta\kappa) > 0$ denote the interior margin from Proposition~\ref{prop:betacrit}. The initial distance and consecutive comparator drift are bounded by half this margin, i.e., $\|u_0 - u^*_0\| \leq R^*/2$ and $\|u^*_t - u^*_{t-1}\| \leq R^*/2$ for every $t$;
    \item[(iv)] $K$ is large enough that $\rho_K := \sqrt{\chi}(1-1/\chi)^{K/2} \leq 1/2$ (suffices to take $K \geq \chi\ln(4\chi)$, where $\chi := L/\mu$).
\end{enumerate}
Then the dynamic regret $\mathrm{DReg}_T := \sum_{t=1}^T [\mathcal{F}_t(u_t) - \mathcal{F}_t(u^*_t)]$ satisfies
\begin{equation}\label{eq:dreg_def}
    \mathrm{DReg}_T
    \;\leq\; \frac{L\rho_K^2\, D_0^2}{1-\rho_K^2}
    \;+\; \frac{L\rho_K^2\, \mathcal{V}_T}{(1-\rho_K)^2\,\kappa^2},
\end{equation}
where $D_0 := \|u_0 - u^*_0\|$ and
$\mathcal{V}_T := \sum_{t=1}^T \sup_u\|\hat{s}_t(u) - \hat{s}_{t-1}(u)\|^2$.
\end{theorem}

Assumptions (i), (ii), and (iv) are standard, namely (i) and (ii) are the per-step regularity and strict-stiffness conditions of Theorem~\ref{thm:contraction} and Proposition~\ref{prop:betacrit}, while (iv) is the usual choice of $K$ making each inner PPC pass at least a factor-of-two contraction. The warm-start margin (iii) is the natural slow-drift condition needed to keep the warm start inside the current attraction ball. Splitting the interior margin $R^*$ into a half-budget for the initial tracking error and a half-budget for the consecutive comparator shift is exactly what the induction at the start of the proof requires to conclude $\|u_{t-1}-u^*_t\|\leq R^*\leq R_t$ at every step, so that Theorem~\ref{thm:contraction}(ii) applies; by Lemma~\ref{lem:comparator}, the drift bound translates to $\sup_u\|\hat{s}_t-\hat{s}_{t-1}\|\leq\kappa R^*/2$, which is the regime in which dynamic regret is informative.

The bound in~\eqref{eq:dreg_def} consists of two terms. The first, $L\rho_K^2 D_0^2/(1-\rho_K^2)$, is a transient that captures the cost of starting away from the equilibrium and vanishes geometrically with the number of inner steps $K$ through $\rho_K$. The second, $L\rho_K^2\mathcal{V}_T/[(1-\rho_K)^2\kappa^2]$, is the steady-state cost of tracking a moving manifold, proportional to the cumulative score variation $\mathcal{V}_T$ and inversely proportional to $\kappa^2$. In classical online control, the analogous drift sensitivity depends on a Lipschitz constant of the known dynamics; here $\kappa$ plays that role despite $f_t$ being unknown, because all manifold information is mediated through the learned density.

\begin{proof}
Define the tracking error $e_t := \|u_t - u^*_t\|$ and comparator shift $d_t := \|u^*_t - u^*_{t-1}\|$. By induction on $t$, hypotheses~(iii)--(iv) keep $\|u_{t-1} - u^*_t\| \leq e_{t-1} + d_t \leq R^*/2 + R^*/2 = R^* \leq R_t$ (using $R_t := \mathrm{dist}(u^*_t,\partial\hat{\mathcal{M}}_t^\alpha) \geq R^*$ from Proposition~\ref{prop:betacrit}), so $u_{t-1}$ lies in the attraction ball of Theorem~\ref{thm:contraction} around $u^*_t$ at every step. Applying Theorem~\ref{thm:contraction}(ii) with $K$ gradient steps on $\mathcal{F}_t$ starting from $u_{t-1}$ gives
\begin{equation}\label{eq:ek}
    e_t \;\leq\; \rho_K\,\|u_{t-1} - u^*_t\| \;\leq\; \rho_K\,(e_{t-1} + d_t),
\end{equation}
where the second inequality is the triangle inequality.

Iterating~\eqref{eq:ek} leads to (using summation index $\sigma$, with $\tau$ reserved for the manifold reach of Theorem~\ref{thm:learning})
\begin{equation}\label{eq:unroll}
    e_t \;\leq\; \rho_K^t\, D_0 \;+\; \sum_{\sigma=1}^t \rho_K^{t-\sigma+1}\, d_\sigma.
\end{equation}

By the Cauchy--Schwarz inequality applied to the weighted sum (with weights $\rho_K^{t-\sigma+1}$ and $\sum_{\sigma=1}^t\rho_K^{t-\sigma+1}\leq\rho_K/(1-\rho_K)$),
\begin{equation}
    \Bigl(\sum_{\sigma=1}^t \rho_K^{t-\sigma+1} d_\sigma\Bigr)^2 \;\leq\; \frac{\rho_K}{1-\rho_K}\sum_{\sigma=1}^t \rho_K^{t-\sigma+1} d_\sigma^2.
\end{equation}
Summing over $t$ and exchanging the order of summation,
\begin{equation}\label{eq:sum_sq}
    \sum_{t=1}^T e_t^2 \;\leq\; \frac{2\rho_K^2 D_0^2}{1-\rho_K^2} \;+\; \frac{2\rho_K^2}{(1-\rho_K)^2}\sum_{t=1}^T d_t^2.
\end{equation}

Finally, by $L$-smoothness of $\mathcal{F}_t$,
$\mathcal{F}_t(u_t) - \mathcal{F}_t(u^*_t) \leq \frac{L}{2}\,e_t^2$.
Substituting Lemma~\ref{lem:comparator} ($d_t \leq \sup_u\|\hat{s}_t - \hat{s}_{t-1}\|/\kappa$) into~\eqref{eq:sum_sq} yields
\begin{equation}
    \mathrm{DReg}_T \leq \frac{L}{2}\sum_{t=1}^T e_t^2 \leq \frac{L\rho_K^2 D_0^2}{1-\rho_K^2} + \frac{L\rho_K^2\,\mathcal{V}_T}{(1-\rho_K)^2\,\kappa^2}. \qedhere
\end{equation}
\end{proof}

The regret is stated in terms of the free energy $\mathcal{F}_t$ because PPC descends on $\mathcal{F}_t$, not on $c$ directly. A cost regret bound follows since for any $u_t, u^*_t \in \hat{\mathcal{M}}_t^\alpha$ the log-density difference is bounded, and hence $\sum_t [c(u_t) - c(u^*_t)]$ differs from $\mathrm{DReg}_T$ by at most an additive term proportional to $\beta$. Safety is treated separately in Theorem~\ref{thm:iss_rigorous}.

The barrier curvature $\kappa$ plays a dual role in~\eqref{eq:dreg_def}. It enters the contraction factor $\rho_K$ through $\mu = \beta\kappa$, so larger $\kappa$ accelerates convergence, and it appears in the drift sensitivity through $1/\kappa^2$, so larger $\kappa$ also reduces the impact of manifold motion. Both effects favor a density with sharp boundaries. In practice, $\kappa$ is a property of the learned density and cannot be tuned directly. The controllable quantity is $\beta_t$, which the adaptive schedule~\eqref{eq:beta_schedule} sets to $\beta^*_{\mathrm{curv}}(1 + C/\sqrt{N_t})$. When $N_t$ is small the schedule inflates $\beta_t$ above the critical threshold, favoring safety over cost optimality. As $N_t$ grows and the density estimate improves, $\beta_t$ decreases toward $\beta^*_{\mathrm{curv}}$, gradually releasing the Planner to exploit more of the feasible space.

\subsection{Online Score Convergence}
\label{sec:online_convergence}

This subsection provides the statistical bridge from the fixed-density analysis in Sections~IV-A and IV-B to the online setting where the KDE estimate $\hat{p}_t$ from~\eqref{eq:kde} evolves with new simulator samples. Algorithms~\ref{alg:ppc}--\ref{alg:nfc} present the controller using the marginal KDE (no context); the convergence guarantees below are stated for this case, and the contextual case follows analogously with $m$ replaced by $m+d$. Lemma~\ref{lem:kde_uniform} controls pointwise density and gradient errors, and Theorem~\ref{thm:learning} converts those controls into the integrated squared score error $\epsilon_t$.

\begin{lemma}[KDE Uniform Convergence]
\label{lem:kde_uniform}
Let $p^*$ be a $C^2$ density on a compact domain $\Omega \subset \mathbb{R}^m$ satisfying $\inf_\Omega p^* > 0$ and $\|p^*\|_{C^2} < \infty$. Let $\hat{p}_N$ be the Gaussian KDE with bandwidth $h = \Theta(N^{-1/(m+4)})$ built from $N$ i.i.d.\ samples from $p^*$. Then with probability at least $1 - \delta$:
\begin{align}
\label{eq:kde_density}
    \|\hat{p}_N - p^*\|_\infty &\;\leq\; C_1\, N^{-2/(m+4)} \sqrt{\ln(N/\delta)}, \\
\label{eq:kde_gradient}
    \|\nabla\hat{p}_N - \nabla p^*\|_\infty &\;\leq\; C_2\, N^{-1/(m+4)} \sqrt{\ln(N/\delta)},
\end{align}
where $C_1, C_2$ depend on $\|p^*\|_{C^2}$, $m$, and $\mathrm{vol}(\Omega)$.
\end{lemma}

\begin{proof}
The KDE is $\hat{p}_N(u) = \frac{1}{N}\sum_{i=1}^N K_h(u - U_i)$ with $K_h(z) = h^{-m}K(z/h)$ and $K$ the standard Gaussian kernel. By Taylor expansion of $p^*$ under convolution, the bias satisfies $\|\mathbb{E}[\hat{p}_N] - p^*\|_\infty = O(h^2)$, and differentiating under the integral gives $\|\nabla\mathbb{E}[\hat{p}_N] - \nabla p^*\|_\infty = O(h^2)$. The pointwise variance of $\hat{p}_N(u)$ is $\mathrm{Var}[\hat{p}_N(u)] \leq C/(Nh^m)$, and for the gradient $\mathrm{Var}[\nabla\hat{p}_N(u)] \leq C/(Nh^{m+2})$, since $\nabla K_h$ introduces an extra $1/h$ factor. A bounded-difference inequality over an $\varepsilon$-net of $\Omega$ (cardinality $O((\mathrm{diam}(\Omega)/\varepsilon)^m)$) with a union bound yields, for any $\delta \in (0,1)$ and with probability at least $1-\delta$,
\[
    \sup_{u \in \Omega}|\hat{p}_N(u) - \mathbb{E}[\hat{p}_N(u)]|
    \leq O\!\left(\sqrt{\tfrac{\ln(N/\delta)}{Nh^m}}\right),
\]
and the analogous bound for $\nabla\hat{p}_N$ with $h^m$ replaced by $h^{m+2}$.

Setting $h = \Theta(N^{-1/(m+4)})$ balances the density bias $O(N^{-2/(m+4)})$ with its standard deviation $O(\sqrt{\ln N}\cdot N^{-2/(m+4)})$, yielding~\eqref{eq:kde_density} (the logarithmic factor is absorbed into $C_1$). For the gradient, the bias is still $O(N^{-2/(m+4)})$ but the standard deviation is $O(\sqrt{\ln N}\cdot N^{-1/(m+4)})$, which dominates~\cite{tsybakov2009introduction}, giving~\eqref{eq:kde_gradient}.
\end{proof}

\begin{theorem}[Score Tracking Bound]
\label{thm:learning}
Let $\hat{p}_t$ be a Gaussian KDE with bandwidth $h = \Theta(N_t^{-1/(m+4)})$ built from $N_t$ cumulative feasibility samples on a manifold with reach $\tau > 0$. Then the integrated squared score error satisfies
\begin{equation}\label{eq:score_rate}
    \mathbb{E}\!\left[\int_{\hat{\mathcal{M}}_t^\alpha} \|\hat{s}_t(u) - s^*_t(u)\|^2\, \hat{p}_t(u)\, du\right] \;\leq\; \frac{C(m,\tau,\alpha)}{N_t^{2/(m+4)}},
\end{equation}
where $C(m,\tau,\alpha)$ depends polynomially on $m$, $1/\tau$, and $1/\alpha$; the $1/\alpha$ dependence enters through the denominator $\hat{p}_t\geq\alpha$ on $\hat{\mathcal{M}}_t^\alpha$ via $p_{\min} := \alpha/2$.
\end{theorem}

\begin{proof}
By the chain rule, $\hat{s}(u) = \nabla_u \ln \hat{p}(u) = \nabla\hat{p}(u)/\hat{p}(u)$ and similarly $s^*(u) = \nabla p^*(u)/p^*(u)$. Subtracting, we obtain
\begin{equation}\label{eq:score_decomp}
    \hat{s}(u) - s^*(u) \;=\; \underbrace{\frac{\nabla(\hat{p} - p^*)(u)}{\hat{p}(u)}}_{=:\, A(u)} \;-\; \underbrace{\frac{\hat{p}(u) - p^*(u)}{\hat{p}(u)}\, s^*(u)}_{=:\, B(u)},
\end{equation}
where the identity follows from a direct computation:
\begin{align*}
    A(u) - B(u)
    &= \frac{\nabla\hat{p} - \nabla p^*}{\hat{p}} - \frac{(\hat{p}-p^*)\,s^*}{\hat{p}} \\
    &= \frac{\nabla\hat{p}\cdot p^* - \nabla p^*\cdot\hat{p}}{\hat{p}\cdot p^*}
    \;=\; \hat{s} - s^*.
\end{align*}

On $\hat{\mathcal{M}}_t^\alpha$ we have $\hat{p}_t \geq \alpha$, and $\|\hat{p}_t - p^*_t\|_\infty = O(N^{-2/(m+4)})$ by Lemma~\ref{lem:kde_uniform}; therefore, for $N$ large enough that $\|\hat{p}_t - p^*_t\|_\infty \leq \alpha/2$, we obtain $p^*_t \geq \alpha/2$ on $\hat{\mathcal{M}}_t^\alpha$. Set $p_{\min} := \alpha/2 > 0$; then $\hat{p}_t \geq p_{\min}$ on $\hat{\mathcal{M}}_t^\alpha$ as well (in fact $\hat{p}_t \geq \alpha = 2p_{\min}$). Note that this is \emph{not} an appeal to the oracle $\inf_{\mathcal{M}_t} p^*_t$ bound of Definition~\ref{def:feasibility_score}, which fails near $\partial\mathcal{M}_t$; the threshold $\alpha > 0$ itself guarantees the bound on the learned set. The score $s^*$ is bounded on the compact set $\hat{\mathcal{M}}_t^\alpha$: $\|s^*\|_\infty \leq S_{\max} < \infty$. Therefore:
\begin{align}
    \|A(u)\| &\leq \frac{2}{p_{\min}}\|\nabla(\hat{p}-p^*)\|_\infty \;\leq\; \frac{2C_2}{p_{\min}}\, N^{-1/(m+4)}, \label{eq:A_bound} \\
    \|B(u)\| &\leq \frac{2S_{\max}}{p_{\min}}\|\hat{p}-p^*\|_\infty \;\leq\; \frac{2C_1 S_{\max}}{p_{\min}}\, N^{-2/(m+4)}. \label{eq:B_bound}
\end{align}

By \eqref{eq:score_decomp} and $\|A-B\|^2 \leq 2\|A\|^2 + 2\|B\|^2$:
\begin{align}
\nonumber
    \int_{\hat{\mathcal{M}}_t^\alpha}\|\hat{s} - s^*\|^2\hat{p}\, du &\leq 2\int_{\hat{\mathcal{M}}_t^\alpha}\bigl(\|A\|^2 + \|B\|^2\bigr)\hat{p}\, du \\
\label{eq:integrated_bound}
    &\leq 2\bigl(\|A\|_\infty^2 + \|B\|_\infty^2\bigr) \;\leq\; \frac{C(m,\tau,\alpha)}{N^{2/(m+4)}},
\end{align}
where $C(m,\tau,\alpha)$ absorbs the constants from \eqref{eq:A_bound}--\eqref{eq:B_bound} and depends on $p_{\min} = \alpha/2$, $S_{\max}$, $\|p^*\|_{C^2}$, and $\mathrm{vol}(\hat{\mathcal{M}}_t^\alpha)$. These quantities depend polynomially on $m$, $1/\tau$ (through the manifold geometry), and $1/\alpha$ (through the $1/p_{\min}^2$ prefactors of~\eqref{eq:A_bound}--\eqref{eq:B_bound}). The bound~\eqref{eq:integrated_bound} holds on the event of Lemma~\ref{lem:kde_uniform}, which has probability at least $1-\delta$. The expectation in~\eqref{eq:score_rate} follows by choosing $\delta = N^{-2/(m+4)}$ and noting that the complementary event contributes at most $O(\delta \cdot S_{\max}^2) = O(N^{-2/(m+4)})$, which is absorbed into $C(m,\tau,\alpha)$.

The rate $N^{-2/(m+4)}$ is the minimax-optimal rate for nonparametric gradient estimation in $m$ dimensions~\cite{tsybakov2009introduction}. The plug-in score estimator inherits this rate because $\hat{s} = \nabla\hat{p}/\hat{p}$ and the gradient error in the numerator dominates the density error in the denominator. For comparison, the density mean integrated squared error (MISE), $\mathbb{E}[\int (\hat{p}_N - p^*)^2\, du]$, converges at the faster rate $N^{-4/(m+4)}$, reflecting the additional difficulty of derivative estimation.
\end{proof}

\begin{remark}
For $m = 2$ (our experiments), the integrated squared score error decays as $N^{-1/3}$, which is slower than the density MISE rate $N^{-2/3}$ but still sufficient for the safety bound of Theorem~\ref{thm:iss_rigorous} to shrink monotonically as $N$ grows (\Cref{fig:sample_budget}).
\end{remark}

\subsection{Safety Bound via Score Descent}

We now compose the contraction guarantee (Theorem~\ref{thm:contraction}) with the statistical bound (Theorem~\ref{thm:learning}) to obtain the main safety result. The bound is analogous to Input-to-State Safety (ISSf)~\cite{ames2019control} in that the distance to the safe set is bounded by a function of the estimation error $\epsilon_t$, but the sensitivity is governed by $\kappa$ rather than a Lipschitz constant of the dynamics. A supporting lemma first controls the set approximation error.

\begin{lemma}[Level-Set Stability]
\label{lem:hausdorff}
Fix a time step $t$. Let $\hat{\mathcal{M}}_t^\alpha := \{u : \hat{p}_t(u \mid \xi_t) \geq \alpha\}$ and let $\epsilon_t$ be the score error of Theorem~\ref{thm:learning}. Assume the oracle density $p^*_t$ (Definition~\ref{def:feasibility_score}) has nondegenerate inward-normal derivative in a boundary tube. Specifically, there exist $c_\partial,r_\partial>0$ with $r_\partial\leq\tau$ such that
\begin{equation}\label{eq:bdry_reg}
    \langle \nabla_u p^*_t(u\mid\xi_t),\,\hat{n}(u)\rangle \;\geq\; c_\partial
\end{equation}
for all $u$ satisfying $\mathrm{dist}(u,\partial\mathcal{M}_t)\leq r_\partial$, where $\hat{n}(u)$ is the inward unit normal at the nearest boundary point. If $\alpha$ and $\sqrt{\epsilon_t}$ are small enough that
\begin{equation}\label{eq:hausdorff_smallness}
    \sqrt{\epsilon_t} < \alpha
    \qquad\text{and}\qquad
    (\alpha+\sqrt{\epsilon_t})/c_\partial \leq r_\partial,
\end{equation}
then
\begin{equation}\label{eq:hausdorff}
    d_H(\hat{\mathcal{M}}_t^\alpha,\,\mathcal{M}_t) \;\leq\; \frac{\alpha + \sqrt{\epsilon_t}}{c_\partial}.
\end{equation}
\end{lemma}

The inward-normal nondegeneracy~\eqref{eq:bdry_reg} is a standard transversality condition that rules out flat-approach boundaries where $\nabla p^*_t$ vanishes; it is what lets density perturbations translate linearly into boundary displacements, and it holds generically for smooth densities whose support has a $C^2$ boundary. The smallness condition~\eqref{eq:hausdorff_smallness} simply requires the level-set threshold $\alpha$ and the sup-norm error $\sqrt{\epsilon_t}$ to stay inside the boundary tube of width $r_\partial$, which is automatic for $N$ large enough since $\epsilon_t\to 0$.

\begin{proof}
Lemma~\ref{lem:kde_uniform} gives $\|\hat{p}_t - p^*_t\|_\infty^2 = O(N^{-4/(m+4)}\ln N)$, which is eventually dominated by $\epsilon_t = O(N^{-2/(m+4)})$; we use the resulting sup-norm bound $\|\hat{p}_t - p^*_t\|_\infty \leq \sqrt{\epsilon_t}$. Write $s_\alpha := (\alpha+\sqrt{\epsilon_t})/c_\partial$; by~\eqref{eq:hausdorff_smallness}, $s_\alpha \leq r_\partial$.

\emph{Step 1: $\hat{\mathcal{M}}_t^\alpha \subseteq \mathcal{M}_t$.} If $u \in \hat{\mathcal{M}}_t^\alpha$, then $p^*_t(u) \geq \hat{p}_t(u) - \sqrt{\epsilon_t} \geq \alpha - \sqrt{\epsilon_t} > 0$ by~\eqref{eq:hausdorff_smallness}, so $u\in\mathcal{M}_t$ (the oracle density vanishes outside $\mathcal{M}_t$). Hence $\sup_{u\in\hat{\mathcal{M}}_t^\alpha}\mathrm{dist}(u,\mathcal{M}_t) = 0$.

\emph{Step 2: every $v\in\mathcal{M}_t$ is within $s_\alpha$ of $\hat{\mathcal{M}}_t^\alpha$.} Let $v\in\mathcal{M}_t$. If $\mathrm{dist}(v,\partial\mathcal{M}_t) \geq s_\alpha$, let $u_b = \pi_{\partial\mathcal{M}_t}(v)$ (well-defined since $s_\alpha \leq r_\partial\leq\tau$) and parametrize the inward normal ray $\gamma(s) = u_b + s\hat{n}(v)$ for $s\in[0,s_\alpha]$, which stays in the $r_\partial$-tube. The fundamental theorem of calculus and~\eqref{eq:bdry_reg} give
\begin{equation}\label{eq:ftc_pstar}
    p^*_t(\gamma(s_\alpha)) \;=\; \int_0^{s_\alpha} \langle\nabla p^*_t(\gamma(s)),\hat{n}(v)\rangle\,ds \;\geq\; c_\partial s_\alpha \;=\; \alpha+\sqrt{\epsilon_t}.
\end{equation}
Because $p^*_t$ is increasing along the inward ray in the tube, monotonicity gives $p^*_t(v)\geq p^*_t(\gamma(s_\alpha)) \geq \alpha+\sqrt{\epsilon_t}$ whenever $\mathrm{dist}(v,\partial\mathcal{M}_t)\geq s_\alpha$, hence $\hat{p}_t(v) \geq p^*_t(v) - \sqrt{\epsilon_t} \geq \alpha$ and $v\in\hat{\mathcal{M}}_t^\alpha$, i.e., $\mathrm{dist}(v,\hat{\mathcal{M}}_t^\alpha) = 0$.

If instead $\mathrm{dist}(v,\partial\mathcal{M}_t) < s_\alpha$, let $u_b = \pi_{\partial\mathcal{M}_t}(v)$ and set $v' := u_b + s_\alpha\hat{n}(v)$. By the same FTC argument, $v'\in\hat{\mathcal{M}}_t^\alpha$, and $\|v - v'\| \leq s_\alpha$, so $\mathrm{dist}(v,\hat{\mathcal{M}}_t^\alpha) \leq s_\alpha$.

Combining Steps 1 and 2, $d_H(\hat{\mathcal{M}}_t^\alpha,\mathcal{M}_t) \leq s_\alpha = (\alpha+\sqrt{\epsilon_t})/c_\partial$.
\end{proof}

Composing the contraction guarantee (Theorem~\ref{thm:contraction}) with the statistical bound (Theorem~\ref{thm:learning}) yields the main safety result, a safety guarantee whose ultimate bound depends on manifold geometry rather than a Lipschitz constant of unknown dynamics. Let $N_t$ denote the cumulative number of feasibility samples available at time $t$, and let $\epsilon_t := C(m,\tau,\alpha)/N_t^{2/(m+4)}$ denote the integrated score-error rate from Theorem~\ref{thm:learning}.

\begin{theorem}[Contextual Safety Bound]
\label{thm:iss_rigorous}
Let $\{u_t\}$ be the output of contextual PPC with stiffness $\beta_t$ and $K$ inner gradient steps warm-started from $u_{t-1}$, and let $\kappa(\xi_t)$ denote the barrier curvature of $\hat{p}_t(\cdot\mid\xi_t)$ (Definition~\ref{def:barrier_curvature}). Assume, for every $t$:
\begin{enumerate}
    \item[(i)] Assumptions~\ref{ass:compactness}--\ref{ass:logconcave} hold for $\hat{p}_t(\cdot\mid\xi_t)$;
    \item[(ii)] $\beta_t > \beta^*_{\mathrm{curv}}$;
    \item[(iii)] the inward-normal nondegeneracy~\eqref{eq:bdry_reg} and smallness conditions~\eqref{eq:hausdorff_smallness} of Lemma~\ref{lem:hausdorff} hold;
    \item[(iv)] the warm-start margin of Theorem~\ref{thm:dynamic_regret}(iii) holds, so that $u_{t-1}$ lies in the attraction ball of Theorem~\ref{thm:contraction} around $u^*_t$;
    \item[(v)] $K$ is large enough that $\rho_K\,\mathrm{diam}(\hat{\mathcal{M}}_t^\alpha)$ is negligible relative to the other terms of~\eqref{eq:iss_bound} (quantified in the proof), where $\rho_K := \sqrt{\chi}\,(1-1/\chi)^{K/2}$ is the iterate contraction factor of Theorem~\ref{thm:contraction}(ii) with condition number $\chi := L/\mu$, $\mu = \beta_t\kappa(\xi_t)$, and $L = L_c + \beta_t\Lambda$.
\end{enumerate}
Then
\begin{equation}\label{eq:iss_bound}
    \mathrm{dist}(u_t, \mathcal{M}_t) \;\leq\;
    \underbrace{\frac{\sqrt{\epsilon_t} + \alpha}{c_\partial}}_{\text{set approx.}}
    \;+\;
    \underbrace{\frac{G_c}{\beta_t\kappa(\xi_t)}}_{\text{cost--safety}},
\end{equation}
where $\epsilon_t := C(m,\tau,\alpha)/N_t^{2/(m+4)}$ is the integrated squared score error from Theorem~\ref{thm:learning}.
\end{theorem}

\begin{remark}[Dependencies and asymptotic tightening]\label{rem:iss_dependencies}
The bound~\eqref{eq:iss_bound} separates a statistical set-approximation term from a deterministic cost--safety term. The score error $\epsilon_t$ depends on the cumulative sample count $N_t$ and the ambient dimension $m$ (or $m+d$ for the conditional KDE), not on $\xi_t$ directly; the context enters the bound through $\kappa(\xi_t)$, which reflects the local geometry of the conditional density. As $N_t\to\infty$, $\epsilon_t\to 0$ and the bound tightens to $\alpha/c_\partial + G_c/(\beta_t\kappa(\xi_t))$.
\end{remark}

\begin{proof}
Let $\bar{u}_t := \arg\max_{u\in\hat{\mathcal{M}}_t^\alpha}\hat{p}_t(u\mid\xi_t)$ denote the maximum-density point of the learned conditional density, which is interior to $\hat{\mathcal{M}}_t^\alpha$ with $\mathrm{dist}(\bar{u}_t,\partial\hat{\mathcal{M}}_t^\alpha)\geq r_\alpha$ (Proposition~\ref{prop:betacrit}, applied to $\hat{p}_t(\cdot\mid\xi_t)$). Decompose the distance via $\bar{u}_t$ using the triangle inequality:
\begin{align}
\nonumber
     \mathrm{dist}(u_t, \mathcal{M}_t)
   \leq \|u_t - u^*_t\| + \|u^*_t - \bar{u}_t\| &+ \mathrm{dist}(\bar{u}_t, \mathcal{M}_t) \\
\label{eq:iss_decomp}
    = \underbrace{\|u_t - u^*_t\|}_{\text{planner residual}}
       + \underbrace{\|u^*_t - \bar{u}_t\|}_{\text{cost--safety offset}}
       &+ \underbrace{\mathrm{dist}(\bar{u}_t, \mathcal{M}_t)}_{\text{set approximation}}.
\end{align}

\emph{Cost--safety offset.} Applying Proposition~\ref{prop:betacrit} to the conditional density $\hat{p}_t(\cdot\mid\xi_t)$, whose barrier curvature is $\kappa(\xi_t)$ by Definition~\ref{def:barrier_curvature}, yields
\begin{equation}\label{eq:betacrit_contextual}
    \|u^*_t - \bar{u}_t\| \;\leq\; \frac{G_c}{\beta_t\,\kappa(\xi_t)},
\end{equation}
and ensures $u^*_t\in\mathrm{int}(\hat{\mathcal{M}}_t^\alpha)$ with $\mathrm{dist}(u^*_t,\partial\hat{\mathcal{M}}_t^\alpha)\geq r_\alpha-G_c/(\beta_t\kappa(\xi_t))>0$ whenever $\beta_t>\beta^*_{\mathrm{curv}}$. This is where the contextual curvature enters the bound, since the $G_c/(\beta\kappa)$ term of Proposition~\ref{prop:betacrit} measures the \emph{offset of the cost-penalized minimizer from the density peak}, not a residual of the gradient iterates.

\emph{Set approximation.} Since $\bar{u}_t\in\hat{\mathcal{M}}_t^\alpha$, Lemma~\ref{lem:hausdorff} gives
\[
    \mathrm{dist}(\bar{u}_t, \mathcal{M}_t) \;\leq\; d_H(\hat{\mathcal{M}}_t^\alpha, \mathcal{M}_t) \;\leq\; \frac{\sqrt{\epsilon_t}+\alpha}{c_\partial}.
\]

\emph{Planner residual.} By Theorem~\ref{thm:contraction}(ii), $K$ inner gradient steps warm-started from $u_{t-1}$ yield
\[
    \|u_t - u^*_t\| \;\leq\; \rho_K\,\|u_{t-1} - u^*_t\|,
    \qquad \rho_K = \sqrt{\chi}\,(1-1/\chi)^{K/2}.
\]
Since $\rho_K$ decays geometrically and $\|u_{t-1}-u^*_t\| \leq D_{\mathcal{M},t} := \mathrm{diam}(\hat{\mathcal{M}}_t^\alpha) < \infty$ (the learned superlevel set is compact as the closed $\alpha$-superlevel set of a Gaussian KDE, which vanishes at infinity; $D_{\mathcal{M},t}$ is a \emph{diameter}, distinct from the initial-distance symbol $D_0=\|u_0-u^*_0\|$ of Theorem~\ref{thm:dynamic_regret}), any tolerance $\delta_K>0$ is reached with $K\geq 2\chi\ln(\sqrt{\chi}\,D_{\mathcal{M},t}/\delta_K)$. Taking $\delta_K$ smaller than the interior margin $r_\alpha-G_c/(\beta_t\kappa(\xi_t))$ also guarantees $u_t\in\hat{\mathcal{M}}_t^\alpha$, so the score-ascent safety filter of Algorithm~\ref{alg:nfc} is inactive in steady state.

Summing the three terms of~\eqref{eq:iss_decomp} and taking $K$ large enough that the planner residual is negligible relative to the other two terms yields~\eqref{eq:iss_bound}. For finite $K$, an additive remainder $\rho_K\,\mathrm{diam}(\hat{\mathcal{M}}_t^\alpha)$ is added to the right-hand side.
\end{proof}

The bound is PPC-specific in that the residual $G_c/(\beta\kappa)$ depends on the barrier curvature $\kappa$ rather than a Lipschitz constant of the unknown dynamics $f_t$. As $N_t$ grows, $\epsilon_t \to 0$ at rate $N_t^{-2/(m+4)}$ (Theorem~\ref{thm:learning}), so the bound tightens to $G_c/(\beta\kappa)$. The adaptive schedule~\eqref{eq:beta_schedule} can then reduce $\beta$ toward $\beta^*_{\mathrm{curv}}$, allowing the Planner to exploit more of the feasible space as estimation improves.

\textbf{The value of context.}
The safety bound~\eqref{eq:iss_bound} depends on the context $\xi_t$ through $\epsilon_t$, $c_\partial$, and $\kappa(\xi_t)$. The first two are properties of the oracle density $p^*_t(\cdot\mid\xi_t)$ (complexity for the score-estimation rate and inward-normal nondegeneracy for the set-approximation constant) that any Planner inherits from the underlying geometry. The only $\xi_t$-dependent quantity the Planner itself \emph{controls} is the barrier curvature of the density it uses in its free energy, namely $\kappa(\xi_t)$ for a context-aware Planner and $\kappa_{\mathrm{marg}} := \inf_{u\in\hat{\mathcal{M}}_t^\alpha}\lambda_{\min}\!\bigl(-\nabla_u^2\ln\hat p_t(u)\bigr)$ (i.e., Definition~\ref{def:barrier_curvature} applied to the marginal density) for a context-blind Planner that replaces $\hat p_t(\cdot\mid\xi_t)$ with the marginal $\hat p_t(u) := \int \hat p_t(u\mid\xi)\,\pi(\xi)\,d\xi$, where $\pi(\xi)$ denotes the \emph{context prior}---the $\xi$-marginal of the joint product-kernel KDE of Section~\ref{sec:kde_instantiation}, i.e., $\pi(\xi) = \tfrac{1}{N_t}\sum_{i=1}^{N_t} K_{h_\xi}(\xi - \xi_i)$---and the distribution underlying every expectation $\mathbb{E}_\pi[\,\cdot\,]$ over contexts below (to be distinguished from the Planner's policy $\pi_t$, which is indexed by time). Since $\kappa$ enters the cost--safety residual $G_c/(\beta_t\kappa)$ inversely, the value of context reduces to comparing these two curvatures. Proposition~\ref{prop:ctx_curvature} records the variational-inference decomposition that underlies the comparison, and Theorem~\ref{thm:ctx_gap} turns it into a quantitative lower bound on the gap between the safety guarantees.

\begin{proposition}[Contextual Curvature Gain]
\label{prop:ctx_curvature}
With $\pi(\xi)$ the context prior defined above, let $w(\xi\mid u) \propto \hat p_t(u\mid\xi)\,\pi(\xi)$ be the posterior over contexts given action $u$, and $\hat s_t(u\mid\xi) := \nabla_u\ln\hat p_t(u\mid\xi)$ the conditional score. The marginal log-density Hessian $\nabla_u^2\ln\hat p_t(u)$ equals
\begin{equation}\label{eq:mixture_hessian}
\mathbb{E}_{w(\xi\mid u)}\!\bigl[\nabla_u^2\ln\hat p_t(u\mid\xi)\bigr]
         \;+\; \mathrm{Cov}_{w(\xi\mid u)}\!\bigl(\hat s_t(u\mid\xi)\bigr).
\end{equation}
\end{proposition}

\begin{proof}
Write $q(\xi) := \hat p_t(u\mid\xi)$, so that $\hat p_t(u) = \mathbb{E}_\pi[q(\xi)]$ and $w(\xi\mid u) = q(\xi)\pi(\xi)/\hat p_t(u)$. Using $\nabla_u q(\xi) = q(\xi)\,\hat s_t(u\mid\xi)$ and differentiating $\ln\hat p_t(u) = \ln\mathbb{E}_\pi[q(\xi)]$ once gives the marginal-score identity
\begin{equation}\label{eq:marg_score}
    \nabla_u\ln\hat p_t(u)
    \;=\; \frac{\mathbb{E}_\pi[\nabla_u q(\xi)]}{\hat p_t(u)}
    \;=\; \mathbb{E}_{w(\xi\mid u)}\!\bigl[\hat s_t(u\mid\xi)\bigr].
\end{equation}
Differentiating $\ln w(\xi\mid u) = \ln q(\xi) + \ln\pi(\xi) - \ln\hat p_t(u)$ in $u$ and substituting~\eqref{eq:marg_score} yields the standard variational-inference identity
\begin{equation}\label{eq:post_logderiv}
    \nabla_u\ln w(\xi\mid u)
    \;=\; \hat s_t(u\mid\xi) - \mathbb{E}_{w(\xi\mid u)}\!\bigl[\hat s_t(u\mid\xi)\bigr],
\end{equation}
i.e., the centered conditional score. Differentiating~\eqref{eq:marg_score} once more and writing the expectation explicitly as $\nabla_u\ln\hat p_t(u) = \int w(\xi\mid u)\,\hat s_t(u\mid\xi)\,d\xi$, the product rule under the integral sign gives
\begin{align*}
    \nabla_u^2\ln\hat p_t(u)
    &= \underbrace{\int [\nabla_u w(\xi\mid u)]\,\hat s_t(u\mid\xi)^{\!\top}\,d\xi}_{=:\,T_1} \\
    &\quad + \underbrace{\int w(\xi\mid u)\,\nabla_u\hat s_t(u\mid\xi)\,d\xi}_{=:\,T_2}.
\end{align*}
For $T_1$, the chain-rule identity $\nabla_u w(\xi\mid u) = w(\xi\mid u)\,\nabla_u\ln w(\xi\mid u)$ combined with the posterior log-derivative~\eqref{eq:post_logderiv} yields $\nabla_u w(\xi\mid u) = w(\xi\mid u)\bigl(\hat s_t(u\mid\xi) - \mathbb{E}_w[\hat s_t(u\mid\xi)]\bigr)$, so
\begin{align*}
    T_1 &= \mathbb{E}_w\!\bigl[\bigl(\hat s_t(u\mid\xi) - \mathbb{E}_w[\hat s_t(u\mid\xi)]\bigr)\,\hat s_t(u\mid\xi)^{\!\top}\bigr] \\
        &= \mathbb{E}_w\!\bigl[\hat s_t(u\mid\xi)\,\hat s_t(u\mid\xi)^{\!\top}\bigr] - \mathbb{E}_w[\hat s_t(u\mid\xi)]\,\mathbb{E}_w[\hat s_t(u\mid\xi)]^{\!\top} \\
        &= \mathrm{Cov}_w\!\bigl(\hat s_t(u\mid\xi)\bigr),
\end{align*}
where the second line pulls the $\xi$-independent factor $\mathbb{E}_w[\hat s_t(u\mid\xi)]$ out of the expectation. For $T_2$, the conditional score satisfies $\nabla_u\hat s_t(u\mid\xi) = \nabla_u^2\ln\hat p_t(u\mid\xi)$ by definition, so $T_2 = \mathbb{E}_w\!\bigl[\nabla_u^2\ln\hat p_t(u\mid\xi)\bigr]$. Adding $T_1$ and $T_2$ yields~\eqref{eq:mixture_hessian}.
\end{proof}

Proposition~\ref{prop:ctx_curvature} is the classical ``total-Hessian'' decomposition: the marginal log-density Hessian equals the posterior-averaged conditional Hessian plus the posterior score covariance. Flipping signs to the positive-semidefinite curvature $-\nabla_u^2\ln\hat p_t \succeq 0$ used elsewhere in the paper, this says that marginalizing \emph{subtracts} $\mathrm{Cov}_w \succeq 0$ from the expected conditional curvature, flattening the barrier pointwise in the L\"owner order; the gap is exactly the posterior variance of the conditional score, which vanishes when the posterior concentrates on a single context and grows as different contexts push the conditional score in substantially different directions. Theorem~\ref{thm:ctx_gap} converts this PSD decomposition into a scalar safety-guarantee gap via Weyl's inequality.

\begin{theorem}[Contextual Safety Gap]
\label{thm:ctx_gap}
Let $B_{\mathrm{ctx}}$ and $B_{\mathrm{marg}}$ denote the Theorem~\ref{thm:iss_rigorous} steady-state safety bounds for the context-aware and context-blind Planners, both sharing the statistical error $\epsilon_t$, boundary constant $c_\partial$, and feasible set $\hat{\mathcal{M}}_t^\alpha$, so that $\kappa_{\mathrm{marg}} := \inf_{u\in\hat{\mathcal{M}}_t^\alpha}\lambda_{\min}(-\nabla_u^2\ln\hat p_t(u))$. Suppose the uniform identifiable-regime condition
\begin{equation}\label{eq:ident_hess}
    \mathbb{E}_{w(\xi\mid u)}\!\bigl[\nabla_u^2\ln\hat p_t(u\mid\xi)\bigr]
    \;=\; \nabla_u^2\ln\hat p_t(u\mid\xi_t)
    \quad \forall\,u\in\hat{\mathcal{M}}_t^\alpha
\end{equation}
holds, and let $\sigma_t^2 := \inf_{u\in\hat{\mathcal{M}}_t^\alpha}\lambda_{\min}\!\bigl(\mathrm{Cov}_{w(\xi\mid u)}(\hat s_t(u\mid\xi))\bigr)\geq 0$. Then
\begin{equation}\label{eq:ctx_gap}
    B_{\mathrm{marg}} - B_{\mathrm{ctx}}
    \;\geq\;
    \frac{G_c\,\sigma_t^2}{\beta_t\,\kappa(\xi_t)\,\kappa_{\mathrm{marg}}}
    \;\geq\; 0,
\end{equation}
with strict inequality whenever $\sigma_t^2 > 0$.
\end{theorem}

\begin{proof}
Theorem~\ref{thm:iss_rigorous} applied separately to the two Planners gives
\[
    B_{\mathrm{ctx}} = \tfrac{\sqrt{\epsilon_t}+\alpha}{c_\partial} + \tfrac{G_c}{\beta_t\kappa(\xi_t)},
    \quad
    B_{\mathrm{marg}} = \tfrac{\sqrt{\epsilon_t}+\alpha}{c_\partial} + \tfrac{G_c}{\beta_t\kappa_{\mathrm{marg}}},
\]
and the common pre-residual cancels on subtraction:
\begin{equation}\label{eq:bound_gap}
    B_{\mathrm{marg}} - B_{\mathrm{ctx}}
    \;=\; \frac{G_c\bigl(\kappa(\xi_t) - \kappa_{\mathrm{marg}}\bigr)}{\beta_t\,\kappa_{\mathrm{marg}}\,\kappa(\xi_t)}.
\end{equation}
It thus suffices to show $\kappa(\xi_t) - \kappa_{\mathrm{marg}} \geq \sigma_t^2$. Fix $u\in\hat{\mathcal{M}}_t^\alpha$ and write $H(u) := -\nabla_u^2\ln\hat p_t(u)$, $\Sigma(u) := \mathrm{Cov}_{w(\xi\mid u)}(\hat s_t(u\mid\xi))\succeq 0$, and $\kappa(\xi_t)(u) := \lambda_{\min}(-\nabla_u^2\ln\hat p_t(u\mid\xi_t))$ for the pointwise context-aware curvature. Proposition~\ref{prop:ctx_curvature} combined with the identifiable-regime condition~\eqref{eq:ident_hess} yields
\[
    H(u) \;=\; -\nabla_u^2\ln\hat p_t(u\mid\xi_t) \;-\; \Sigma(u).
\]
Weyl's inequality for Hermitian matrices states $\lambda_{\min}(A+B) \leq \lambda_{\min}(A) + \lambda_{\max}(B)$. Applying it with $A = -\nabla_u^2\ln\hat p_t(u\mid\xi_t)$ and $B = -\Sigma(u)$, and using $\lambda_{\max}(-\Sigma(u)) = -\lambda_{\min}(\Sigma(u))$ together with $\lambda_{\min}(\Sigma(u)) \geq \sigma_t^2$,
\begin{equation}\label{eq:weyl_pointwise}
    \lambda_{\min}(H(u))
    \;\leq\; \kappa(\xi_t)(u) - \lambda_{\min}(\Sigma(u))
    \;\leq\; \kappa(\xi_t)(u) - \sigma_t^2.
\end{equation}
Taking the infimum over $u\in\hat{\mathcal{M}}_t^\alpha$ of both sides of~\eqref{eq:weyl_pointwise}, and using $\kappa_{\mathrm{marg}} = \inf_u\lambda_{\min}(H(u))$ and $\kappa(\xi_t) = \inf_u\kappa(\xi_t)(u)$,
\[
    \kappa_{\mathrm{marg}} \;\leq\; \kappa(\xi_t) - \sigma_t^2,
\]
i.e., $\kappa(\xi_t) - \kappa_{\mathrm{marg}} \geq \sigma_t^2$. Substituting into~\eqref{eq:bound_gap} yields~\eqref{eq:ctx_gap}; strictness is immediate from $\sigma_t^2 > 0$.
\end{proof}

Theorem~\ref{thm:ctx_gap} turns the informal intuition behind Problem~$\mathcal{P}$ into a quantitative lower bound on the safety-guarantee gap: the context-aware Planner's Theorem~\ref{thm:iss_rigorous} bound improves over the context-blind one by at least $G_c\,\sigma_t^2/(\beta_t\kappa(\xi_t)\kappa_{\mathrm{marg}})$, where $\sigma_t^2$ is the uniform infimum (over the feasible set) of the smallest posterior score-covariance eigenvalue. This gap vanishes when the posterior $w$ collapses to a single context, and is largest when different contexts place the Planner at substantially different conditional scores, which is exactly the regime in which the contextual advantage is most visible in~\Cref{fig:contextual}.

\begin{figure*}[t]
\centering
\begin{tikzpicture}[
    block/.style={draw, rounded corners, minimum height=1.0cm, minimum width=2.1cm, align=center, font=\footnotesize},
    arr/.style={-{Stealth[length=2.5mm]}, thick},
    lbl/.style={font=\scriptsize, fill=white, inner sep=1pt},
]
\node[block, fill=blue!8] (env) {Environment\\[-1pt]$x_{t+1}{=}f_t(x_t,u_t)$};
\node[block, fill=blue!8, right=1.75cm of env] (ora) {World Model\\[-1pt]$\mathcal{W}_t$};
\node[block, fill=blue!8, right=1.75cm of ora] (kde) {KDE\\[-1pt]$\hat{p}_t,\hat{s}_t$};
\node[block, fill=orange!10, right=1.75cm of kde] (ppc) {PPC\\[-1pt]$K$ grad steps};
\node[block, fill=blue!8, right=1.75cm of ppc] (flt) {Safety\\[-1pt]Filter};

\draw[arr] (env) -- node[lbl,below] {$x_t$} (ora);
\draw[arr] (ora) -- node[lbl,below] {$\{u_i\}$} (kde);
\draw[arr] (kde) -- node[lbl,below] {$\hat{s}_t,\beta_t$} (ppc);
\draw[arr] (ppc) -- node[lbl,below] {$u_{\mathrm{ppc}}$} (flt);
\draw[arr, rounded corners=5pt] (flt.south) -- ++(0,-0.45) -| node[lbl,below,pos=0.25] {$u_t$} (env.south);

\node[draw, dashed, rounded corners, fit=(ora)(kde)(flt), inner sep=7pt] (simfit) {};
\node[font=\footnotesize, anchor=south west, inner sep=0pt]
    at ([yshift=3pt]simfit.north west) {\textbf{Simulator} (Alg.~\ref{alg:nfc})};
\node[draw, dashed, rounded corners, fit=(ppc), inner sep=7pt] (plfit) {};
\node[font=\footnotesize, anchor=south, inner sep=0pt]
    at ([yshift=3pt]plfit.north) {\textbf{Planner} (Alg.~\ref{alg:ppc})};
\end{tikzpicture}
\caption{Per-step information flow (Algorithms~\ref{alg:ppc}--\ref{alg:nfc}). The Simulator draws $N$ feasibility samples, fits a KDE to produce $(\hat{s}_t,\beta_t)$, and after receiving the Planner's candidate $u_{\mathrm{ppc}}$, applies a safety filter (score-ascent retraction if $\hat{p}_t(u_{\mathrm{ppc}}){<}\alpha$) before returning $u_t$ to the environment.}
\label{fig:alg_flow}
\end{figure*}
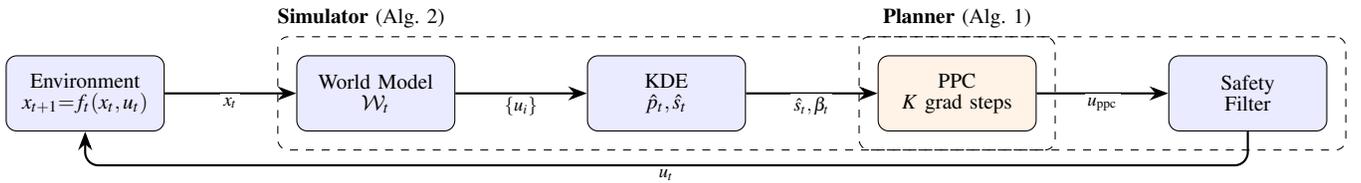

\subsection{Algorithms}

The control law is realized through two interacting components whose information flow is depicted in~\Cref{fig:alg_flow}. The information sets $\mathcal{I}_t^{\mathrm{Sim}}$ and $\mathcal{I}_t^{\mathrm{Pl}}$ defined in~\eqref{eq:info_sim}--\eqref{eq:info_pl} make the asymmetry concrete. The Simulator's algorithm may read $x_t$, $\mathcal{W}_t$, and $\mathcal{M}_t$, whereas the Planner's algorithm may use only $\hat{s}_t$ and $c$. At each time step $t$, the Simulator (Algorithm~\ref{alg:nfc}) draws $N$ feasibility samples from the oracle $\mathcal{W}_t$, updates the density model $\hat{p}_t$, and transmits the score function $\hat{s}_t$ and the stiffness parameter $\beta_t$ to the Planner. The Planner (Algorithm~\ref{alg:ppc}) then performs $K$ gradient descent steps on the PPC free energy using $\hat{s}_t$, and returns the candidate action $u_{\mathrm{ppc}}$ to the Simulator for safety projection. If the candidate lies outside $\hat{\mathcal{M}}_t^\alpha$, the Simulator retracts it onto the learned manifold via score ascent before applying it to the plant.

\begin{algorithm}[!h]
\caption{Planner: Score-Based PPC}
\label{alg:ppc}
\begin{algorithmic}[1]
\REQUIRE Score function $\hat{s}_t(\cdot)$ from Simulator, cost $c(\cdot)$, stiffness $\beta_t$, step size $\eta$, inner steps $K$, previous action $u_{t-1}$.
\ENSURE Candidate action $u_{\mathrm{ppc}}$.

\STATE $u \leftarrow u_{t-1}$ \hfill \textcolor{teal}{// Warm start}
\FOR{$k = 1, \ldots, K$}
    \STATE $u \leftarrow u - \eta \bigl(\nabla_u c(u) - \beta_t\, \hat{s}_t(u)\bigr)$
\ENDFOR
\RETURN $u_{\mathrm{ppc}} \leftarrow u$
\end{algorithmic}
\end{algorithm}

\begin{algorithm}[!h]
\caption{Simulator: Score Oracle and Safety Filter}
\label{alg:nfc}
\begin{algorithmic}[1]
\REQUIRE State $x_t$, feasibility oracle $\mathcal{W}_t$, density model $\hat{p}_{t-1}$, candidate $u_{\mathrm{ppc}}$ from Planner, threshold $\alpha$, retraction steps $J$, retraction rate $\eta_r$, stiffness-schedule constant $C>0$ (\Cref{eq:beta_schedule}).
\ENSURE Safe action $u_t \in \hat{\mathcal{M}}_t^\alpha$, updated score $\hat{s}_t(\cdot)$, stiffness $\beta_t$.

\STATE Draw $N$ feasibility samples $\{u_i\} \sim \mathcal{W}_t(x_t, \cdot)$
\STATE Update density model: $\hat{p}_t \leftarrow \mathrm{KDE}(\{u_i\} \cup \text{buffer})$
\STATE Compute score $\hat{s}_t(u) = \nabla_u \ln \hat{p}_t(u)$
\STATE Set $\beta_t = \beta^*_{\mathrm{curv}}\bigl(1 + C/\sqrt{N_t}\bigr)$
\STATE \textbf{Transmit} $\hat{s}_t(\cdot)$ and $\beta_t$ to Planner
\STATE Receive $u_{\mathrm{ppc}}$ from Planner
\IF{$\hat{p}_t(u_{\mathrm{ppc}}) \geq \alpha$}
    \STATE $u_t \leftarrow u_{\mathrm{ppc}}$
\ELSE
    \STATE $u_t \leftarrow u_{\mathrm{ppc}}$
    \FOR{$j = 1, \ldots, J$}
        \STATE $u_t \leftarrow u_t + \eta_r\, \hat{s}_t(u_t)$ \hfill \textcolor{teal}{// Score ascent}
        \IF{$\hat{p}_t(u_t) \geq \alpha$}
            \STATE \textbf{break}
        \ENDIF
    \ENDFOR
\ENDIF
\STATE Apply $u_t$ to plant
\end{algorithmic}
\end{algorithm}

\subsection{Computational Properties}

\subsubsection{Compositionality}
For composite manifolds $\mathcal{M} = \bigcap_i \mathcal{M}_i$, the total barrier decomposes additively and $\hat{s}_{\mathrm{total}} = \sum_i \hat{s}_i$, enabling modular fusion of heterogeneous constraints.

\subsubsection{Scalability}
Each PPC gradient step requires evaluating the KDE score $\hat{s}_t(u)$ at a single point, which costs $O(Nm)$ where $N$ is the sample count and $m$ is the control dimension, plus an $O(m)$ gradient update. In contrast, CBF-QP methods~\cite{ames2016control} must form and solve a quadratic program whose per-step cost is $O(m^3)$ (dominated by the dense linear system solve), with additional overhead when many constraints are active. The first-order structure of PPC avoids this cubic dependence on $m$, and the score evaluation is embarrassingly parallel across samples.

\subsubsection{Online Refinement}
Theorem~\ref{thm:learning} gives the integrated score error rate $\epsilon_t = O(N_t^{-2/(m+4)})$, so the set-approximation term of the safety bound in Theorem~\ref{thm:iss_rigorous} shrinks as the cumulative sample count $N_t$ grows. The decreasing statistical error then allows the Planner to drive the stiffness $\beta_t$ toward its critical value $\beta^*_{\mathrm{curv}}$, approached from above so that the strict inequality required by Proposition~\ref{prop:betacrit} continues to hold. Lowering $\beta_t$ relaxes the score penalty and lets the Planner exploit a larger portion of the manifold interior, trading conservative safety margins for improved task performance once the data have certified the tighter bound.

\begin{figure*}[t]
    \centering
    \includegraphics[width=\textwidth]{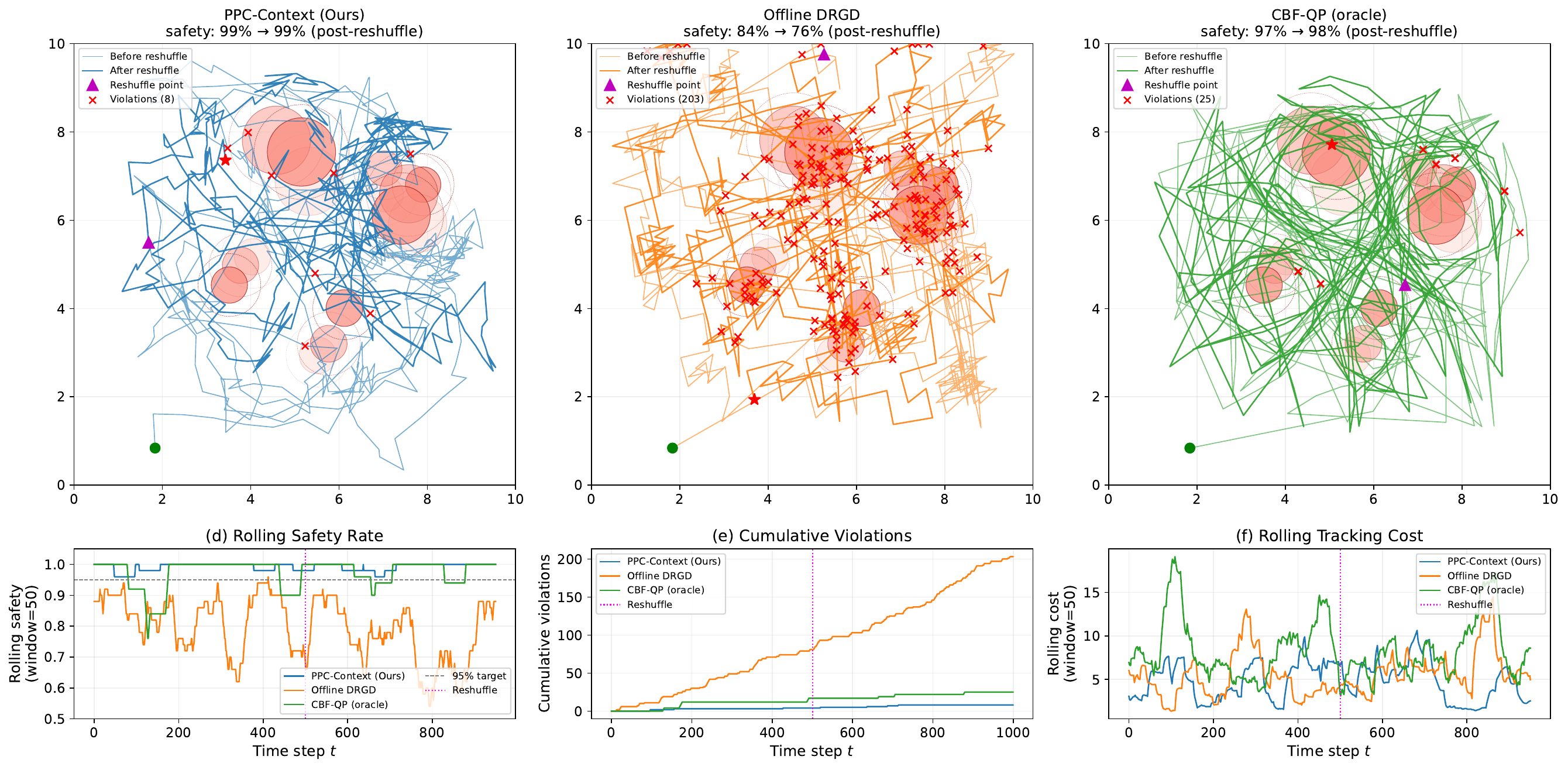}
    \caption{Contextual adaptation under obstacle reshuffle ($T{=}1000$, $N{=}300$ samples/step, reshuffle at $t{=}T/2$). \textbf{Top row:} Trajectories colored by phase (lighter before reshuffle, darker after); red~$\times$ markers indicate safety violations. \textbf{Left:}~PPC-Context~(Ours) leverages the context signal $\xi_t$ to rapidly adapt its density model after the manifold shift, achieving safety competitive with the oracle baseline using only black-box samples. \textbf{Center:}~Offline DRGD uses a frozen score pretrained on $500$ samples at $t{=}0$; after the reshuffle its stale model fails to track the new manifold, producing frequent violations. \textbf{Right:}~CBF-QP~(oracle) has direct access to exact obstacle positions and radii at every step, which is an information advantage unavailable in the black-box setting. \textbf{Bottom row:} (d)~Rolling safety rate with reshuffle marked (magenta line); (e)~cumulative violations; (f)~rolling tracking cost. PPC-Context approaches oracle-level safety using only simulator samples, while offline pretraining without online updates collapses after the shift.}
    \label{fig:trajectories}
\end{figure*}

\section{Experiments}
\label{sec:experiments}

We validate the online score-based framework on a 2D robot navigation task with dynamic obstacles, designed to exercise the core challenges addressed by this paper, namely a \textit{non-convex}, \textit{time-varying} feasibility manifold that is \textit{unknown} to the central planner and must be learned online from black-box interactions.

\begin{figure*}[t]
    \centering
    \includegraphics[width=0.96\textwidth]{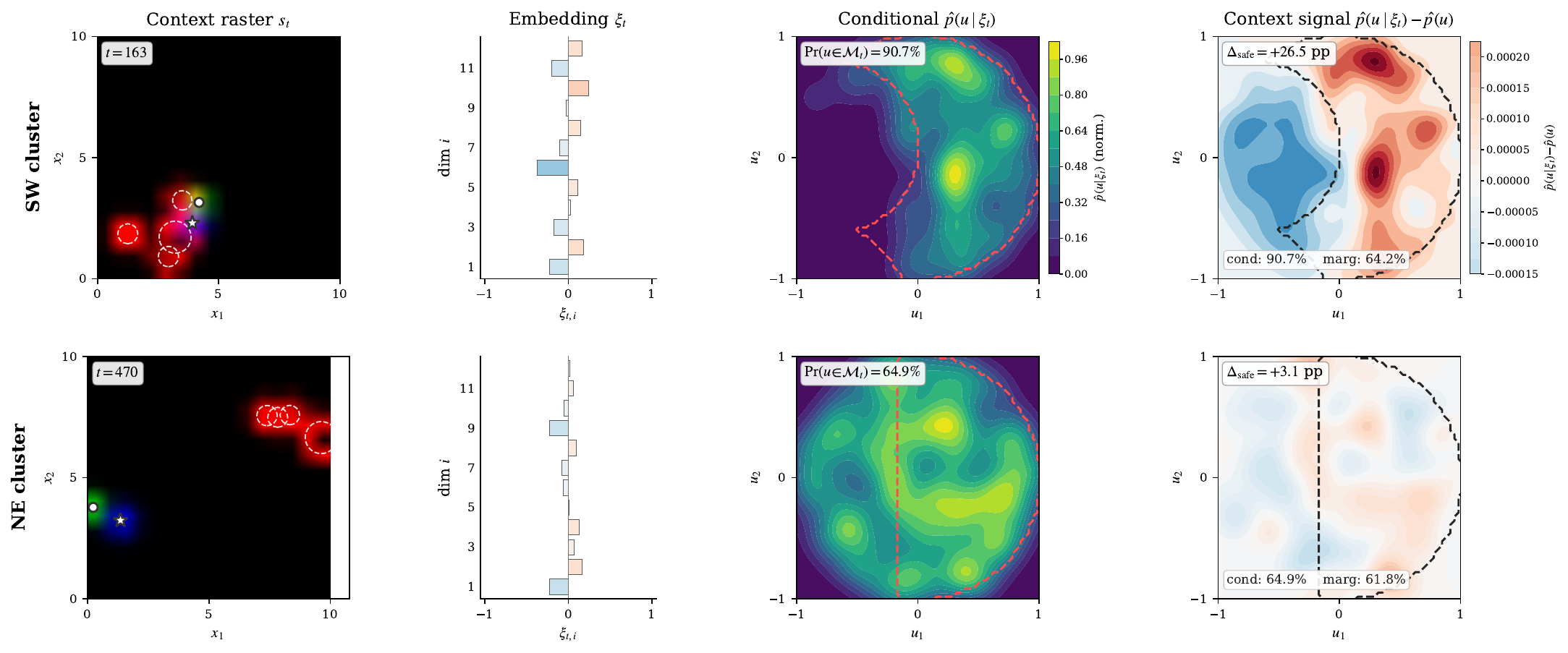}
    \caption{Contextual observation pipeline at the most constraining warmed-up step of two structurally opposite obstacle modes (top: SW cluster; bottom: NE cluster), out of the four quadrant-clustered modes the experiment cycles through every $40$ steps. \textbf{Column~1:} three-channel raster $s_t$ (red~=~obstacles, green~=~robot, blue~=~goal). \textbf{Column~2:} embedding $\xi_t = \tanh(R\,\mathrm{vec}(s_t))$; the two rows carry visibly distinct patterns, so the context kernel can separate the modes. \textbf{Column~3:} conditional density $\hat p(u \mid \xi_t)$ over the admissible action disk, with feasibility boundary $\partial\mathcal{M}_t$ overlaid in red dashed; annotations give the conditional feasibility-mass. \textbf{Column~4:} \emph{context signal} $\hat p(u \mid \xi_t) - \hat p(u)$---red adds mass, blue suppresses it. The key quantity $\Delta_{\mathrm{safe}} = \Pr(u\!\in\!\mathcal{M}_t\mid \hat p(\cdot\mid\xi_t)) - \Pr(u\!\in\!\mathcal{M}_t\mid \hat p(\cdot))$ reported above each Column~4 panel measures the \emph{directional} value of context---the excess feasibility-mass the conditional concentrates inside $\mathcal{M}_t$ relative to the pooled marginal---and is the empirical counterpart of the strictly positive $\sigma_t^2$ that drives the safety gap in Theorem~\ref{thm:ctx_gap}.}
    \label{fig:context_vis}
\end{figure*}

\subsection{Problem Setup}
\label{sec:exp_setup}

A robot navigates a bounded workspace $[0, 10]^2$ to reach a goal position while avoiding $5$ moving circular obstacles. The implementation uses single-integrator dynamics $q_{t+1} = q_t + u_t$ with $\|u_t\| \leq u_{\max} = 1.0$, which yields a non-convex feasible action set at each step while keeping the planning problem two-dimensional; the same interface extends to higher-order dynamics by letting $u_t$ denote the first control block passed to the Simulator. Obstacle $k$ has radius $r_k \in [0.4, 0.8]$ and its center $o_k(t) \in \mathbb{R}^2$ follows a Lissajous curve:
\begin{equation}
\label{eq:obstacle_motion}
    o_k(t) = c_k + \begin{bmatrix} A_k \sin(\omega_k t + \phi_k) \\ B_k \cos(\nu_k t + \psi_k) \end{bmatrix},
\end{equation}
where $c_k$ is a base position and $(A_k, B_k, \omega_k, \nu_k, \phi_k, \psi_k)$ are drawn randomly per seed.

In the Simulator--Planner architecture of~\Cref{sec:architecture}, the Simulator (physics engine) has access to the obstacle states and dynamics, while the Planner's task cost uses the robot position and goal. For the main experiments, the Planner's density model is trained from the same black-box feasibility oracle as our method (no explicit obstacle parameters in the optimization). The Planner minimizes the tracking cost
\begin{equation}
\label{eq:tracking_cost}
    c(q_t, u) = \bigl\| q_t + u - g_t \bigr\|^2,
\end{equation}
where $g_t$ is the current goal position (updated upon arrival). The \textbf{feasibility manifold} is
\begin{equation}
\label{eq:robot_manifold}
    \mathcal{M}_t(q_t) = \bigl\{ u \in \mathcal{U} : \| q_t + u - o_k(t{+}1) \| \geq r_k + d_{\mathrm{safe}},\; \forall k \bigr\},
\end{equation}
with safety margin $d_{\mathrm{safe}} = 0.3$. This set is non-convex in $u$ and its topology changes as obstacles move.

\subsection{Baselines}
\label{sec:baselines}

We compare our online PPC controller (Algorithm~\ref{alg:ppc}) against five baselines spanning offline generative models, data-driven safety methods, sampling-based control, and conservative approximation.

\subsubsection{Offline DRGD~\cite{kharitenko2025landing}}
A KDE score model is fit to $500$ feasibility samples from the \textit{initial} obstacle configuration ($t{=}0$) only. At deployment, the density is fixed and used for manifold-constrained optimization without online updates. This baseline isolates the cost of \textit{not} adapting the manifold model online.

\subsubsection{CBF-QP~\cite{ames2019control}}
A Control Barrier Function $h_k(x) = \|q - o_k(t)\|^2 - (r_k + d_{\mathrm{safe}})^2$ with \textit{oracle access} to obstacle positions projects the nominal input onto the safe set via a QP (class-$\mathcal{K}$ coefficient $\gamma{=}0.5$). This represents the strongest structural prior.

\subsubsection{GP-CBF~\cite{cheng2019end}}
A Gaussian Process models the constraint function $h(x, u) = \min_k (\|q_{t+1}(u) - o_k(t{+}1)\| - r_k - d_{\mathrm{safe}})$ online; its posterior mean serves as a learned CBF within a QP filter (RBF kernel, hyperparameters refitted every 50 steps).

\subsubsection{Sampling-Based MPC (CEM)}
The Cross-Entropy Method~\cite{rubinstein1999cross} draws $300$ candidate actions from a Gaussian proposal and selects the lowest-cost feasible candidate (KDE log-probability above $\alpha$), mirroring CEM-MPC with a learned model~\cite{chua2018deep,pinneri2021sample}. Like PPC, CEM collects $N{=}300$ oracle feasibility samples per step (via rejection sampling) to update its KDE density model; feasibility of CEM's own candidate actions is then assessed against this KDE rather than by re-querying the oracle. The proposal is updated via the elite set (top 10\%, 5 iterations).

\subsubsection{Static Conservative}
The reachable set is conservatively inner-approximated using the maximum swept obstacle radii $r_k^{\max} = r_k + d_{\mathrm{safe}} + \max_t \|o_k(t) - o_k(0)\|$ computed from worst-case displacement. The robot treats these inflated obstacles as static and solves a QP at each step. Because the inflation depends on the randomly drawn Lissajous amplitudes, the feasible set size varies across seeds, leading to higher cost variance than the other baselines.

\subsection{Our Method: Online PPC}
\label{sec:our_method}

The Simulator maintains a kernel density estimate over feasibility samples. At each time step, it draws $N{=}300$ feasible actions via rejection sampling against the black-box oracle for~\eqref{eq:robot_manifold}. The score $\nabla_u \ln \hat{p}(u)$ is evaluated in closed form from the KDE. The critical stiffness $\beta^*_{\mathrm{curv}} = G_c/(\kappa\, r_\alpha)$ is estimated online from the KDE's barrier curvature at the $\alpha$-level-set boundary, and the stiffness schedule follows
\begin{equation}
\label{eq:beta_schedule}
    \beta_t = \beta^*_{\mathrm{curv}} \cdot \bigl(1 + C / \sqrt{N_t}\bigr),
\end{equation}
where $N_t$ is the cumulative number of feasibility samples and $C > 0$ is a tunable constant. The gradient step size is set to $\eta = \min(\eta_0,\, 1/(L_c + \beta_t \Lambda))$ following Theorem~\ref{thm:contraction}, where $\eta_0 = 0.02$ is a base learning rate, $\Lambda = 1/h^2$, and $h$ is the KDE bandwidth.

\textbf{Contextual observation model.}
For the contextual experiment (\Cref{fig:contextual}), the Planner receives $\xi_t \in \mathbb{R}^{12}$ obtained by flattening three top-down raster channels at resolution $16{\times}16$ (obstacle occupancy, robot, goal) and applying a fixed random linear projection with $\tanh$ squashing. The Simulator fits $\hat{p}(u \mid \xi_t)$ with a product-kernel KDE over pairs $(u, \xi)$ from online samples. \Cref{fig:context_vis} illustrates the pipeline on two structurally opposite obstacle modes (SW and NE clusters) and quantifies the safety advantage of conditioning: at the most constraining warmed-up step in each mode, the conditional concentrates $90.7\%$ (SW) and $64.9\%$ (NE) of its mass inside the feasibility manifold $\mathcal{M}_t$, against only $64.2\%$ and $61.8\%$ for the context-blind marginal; the resulting $\Delta_{\mathrm{safe}}$ values of $+26.5$ pp and $+3.1$ pp are the empirical counterpart of the strictly positive $\sigma_t^2$ that drives the safety gap in Theorem~\ref{thm:ctx_gap}.

\subsection{Evaluation Metrics}
\label{sec:metrics}

\begin{enumerate}
    \item \textbf{Safety Rate}: Fraction of time steps where $u_t \in \mathcal{M}_t$. By Theorem~\ref{thm:iss_rigorous}, $\mathrm{dist}(u_t,\mathcal{M}_t)$ is bounded; when that bound is zero (i.e., sufficient samples and $\beta_t > \beta^*_{\mathrm{curv}}$), the action lies inside $\mathcal{M}_t$ and the safety rate approaches $1$.
    \item \textbf{Normalized Cost}: Total tracking cost divided by the Oracle cost. The Oracle minimizes the single-step tracking cost $c(u)$ subject to $u \in \mathcal{M}_t$ with full knowledge of the obstacle geometry at each $t$ (no look-ahead). Values near $1.0$ indicate near-optimal per-step performance.
    \item \textbf{Adaptation Speed}: Number of steps to recover a safety rate $\geq 95\%$ after a sudden obstacle configuration change at $t = 500$ (all Lissajous parameters are re-randomized).
    \item \textbf{Learning Convergence}: Empirical integrated squared score error $\|\hat{s}-s^*\|^2$ vs.\ cumulative samples $N_t$, validating the $\mathcal{O}(N^{-2/(m+4)})$ rate of Theorem~\ref{thm:learning} ($= \mathcal{O}(N^{-1/3})$ for $m = 2$).
\end{enumerate}

\begin{figure*}[t]
    \centering
\includegraphics[width=0.9\textwidth]{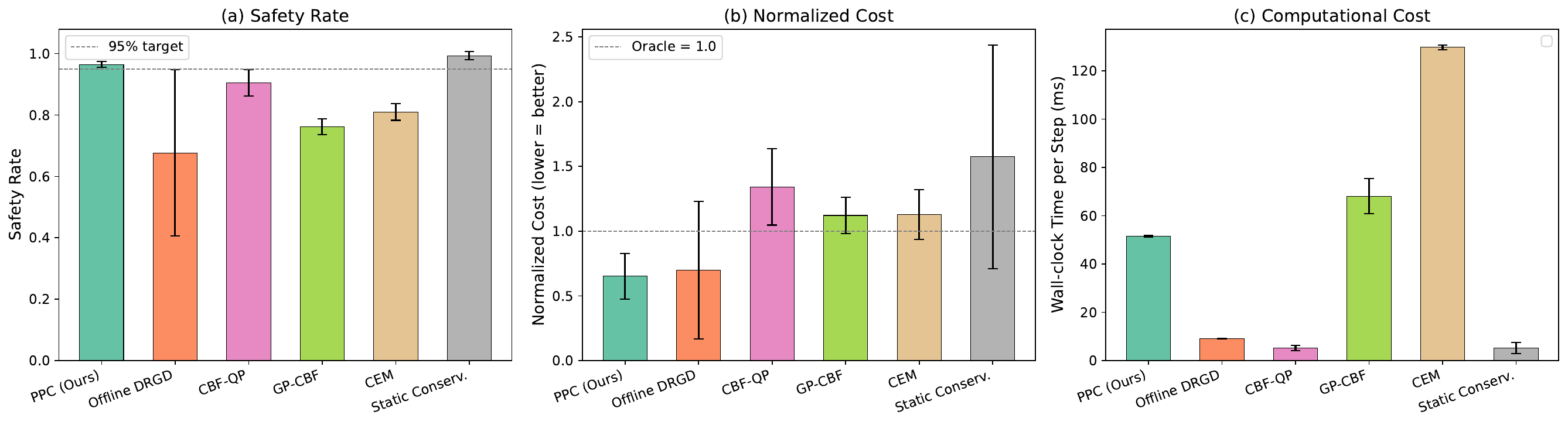}
    \caption{Main comparison across all methods ($T{=}1000$, $N{=}300$ feasibility samples per step, five seeds, mean $\pm$ std); see~\Cref{tab:main_results}. The trajectory figure (\Cref{fig:trajectories}) uses the same per-step sample budget. \textbf{(a)}~Safety rate. \textbf{(b)}~Normalized cost (Oracle $= 1.0$). \textbf{(c)}~Wall-clock time per step.}
    \label{fig:main_comparison}
\end{figure*}

\subsection{Experimental Protocol}
\label{sec:protocol}

\subsubsection{Experiment 1: Main Comparison}
All six methods including PPC and five baselines are run for $T{=}1000$ steps with $N{=}300$ black-box feasibility samples per step (Offline DRGD pretrains on $500$ samples at $t{=}0$ only), across five random seeds (random obstacle trajectories and initial positions). \Cref{tab:main_results} lists mean $\pm$ standard deviation of safety rate and normalized tracking cost; \Cref{fig:main_comparison} shows the same statistics as bar plots. Representative trajectories appear in~\Cref{fig:trajectories}.

\begin{table}[t]
    \caption{Main comparison ($T{=}1000$, $N{=}300$ samples/step, five seeds, mean $\pm$ std). PPC achieves the best safety among black-box methods while maintaining the lowest cost. $^\dagger$The high cost variance of Static Conserv.\ reflects the sensitivity of episode-wide worst-case inflation to random obstacle motion amplitudes; Offline DRGD's variance is similarly driven by the mismatch between its frozen model and the current manifold.}
    \label{tab:main_results}
    \centering
    {\footnotesize
    \begin{tabular}{@{}lcc@{}}
        \toprule
        Method & Safety & Norm.\ cost \\
        \midrule
        PPC (Ours) & $.964 \pm .009$ & $0.65 \pm .18$ \\
        Offline DRGD & $.676 \pm .271$ & $0.70 \pm .53$ \\
        CBF-QP & $.905 \pm .043$ & $1.34 \pm .30$ \\
        GP-CBF & $.762 \pm .025$ & $1.12 \pm .14$ \\
        CEM & $.810 \pm .027$ & $1.13 \pm .19$ \\
        Static Conserv. & $.993 \pm .014$ & $1.57 \pm .86$\rlap{$^\dagger$} \\
        \bottomrule
    \end{tabular}}
\end{table}

\begin{figure}[t]
    \centering
    \includegraphics[width=0.9\columnwidth]{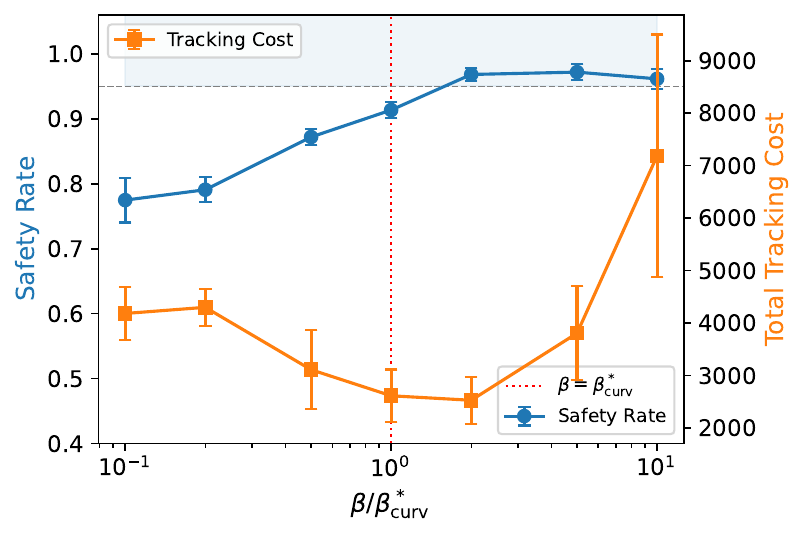}
    \caption{Stiffness ablation validating Proposition~\ref{prop:betacrit}. The critical stiffness $\beta^*_{\mathrm{curv}} = G_c/(\kappa\, r_\alpha)$ is computed from the formula (median over seeds). Safety rate (blue, left axis) rises from $\sim\!76\%$ to $\sim\!96\%$ and crosses the $95\%$ threshold near $\beta = \beta^*_{\mathrm{curv}}$ (red dashed), confirming the predicted phase transition. Tracking cost (orange, right axis) is low for small $\beta$ but rises sharply beyond $\beta^*_{\mathrm{curv}}$, illustrating the cost-safety trade-off governed by barrier curvature.}
    \label{fig:stiffness}
\end{figure}

\subsubsection{Experiment 2: Stiffness Ablation}
The environment is fixed and $\beta$ is varied from $0.1 \beta^*_{\mathrm{curv}}$ to $10 \beta^*_{\mathrm{curv}}$. We plot safety rate and normalized cost as functions of $\beta / \beta^*_{\mathrm{curv}}$, validating the critical stiffness condition of Proposition~\ref{prop:betacrit}: safety should emerge at $\beta \approx \beta^*_{\mathrm{curv}}$. Results are shown in~\Cref{fig:stiffness}.

\subsubsection{Free Energy Landscape}
\Cref{fig:landscape} complements the stiffness sweep with a spatial analysis of the PPC free energy at a representative time step ($t{=}200$, $N{=}200$ samples/step). Panel~(a) shows the tracking cost $c(u)$ with the true feasibility boundary (blue solid) and the learned level set $\hat{\mathcal{M}}_t^\alpha$ (cyan dashed). Panel~(b) overlays the free-energy contours with two key points: the PPC equilibrium $u^*$ (red triangle) and the density maximizer $\bar{u} = \arg\max \hat{p}$ (cyan square). The equilibrium $u^*$ is pulled toward the cost minimum but anchored near $\bar{u}$ by the barrier penalty, exactly as Proposition~\ref{prop:betacrit} predicts. Panel~(c) reports the geometric gaps: $\|u^*{-}\bar{u}\|$ is well within the theoretical bound $G_c/(\beta\kappa)$, and $\mathrm{dist}(u^*, \partial\mathcal{M}_t) > 0$ confirms that $u^*$ lies safely inside the true manifold.

\begin{figure*}[t]
    \centering
\includegraphics[width=0.9\textwidth]{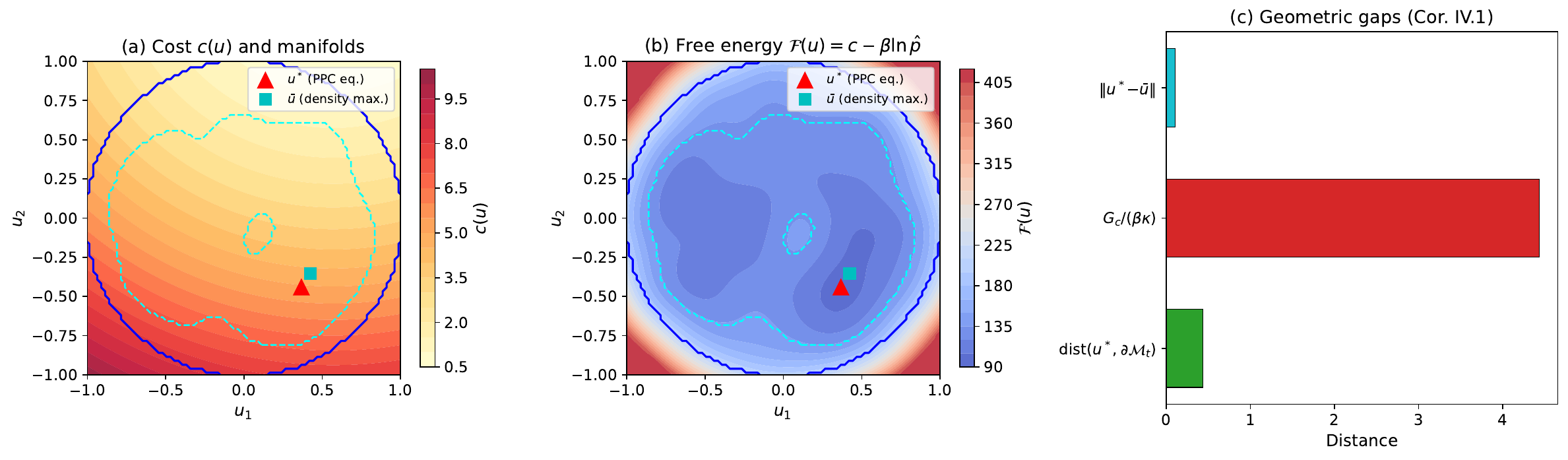}
    \caption{Free-energy landscape at $t{=}200$. \textbf{(a)}~Cost $c(u)$ with true manifold boundary (blue) and learned $\hat{\mathcal{M}}_t^\alpha$ (cyan dashed). \textbf{(b)}~Free energy $\mathcal{F}(u)$ with the PPC equilibrium $u^*$ (red) and density maximizer $\bar{u}$ (cyan). \textbf{(c)}~Geometric gaps: the empirical $\|u^*{-}\bar{u}\|$ is well within the theoretical bound $G_c/(\beta\kappa)$ from Proposition~\ref{prop:betacrit}, and $\mathrm{dist}(u^*, \partial\mathcal{M}_t) > 0$ confirms the equilibrium is safely interior.}
    \label{fig:landscape}
\end{figure*}

\begin{figure}[t]
    \centering
\includegraphics[width=0.9\columnwidth]{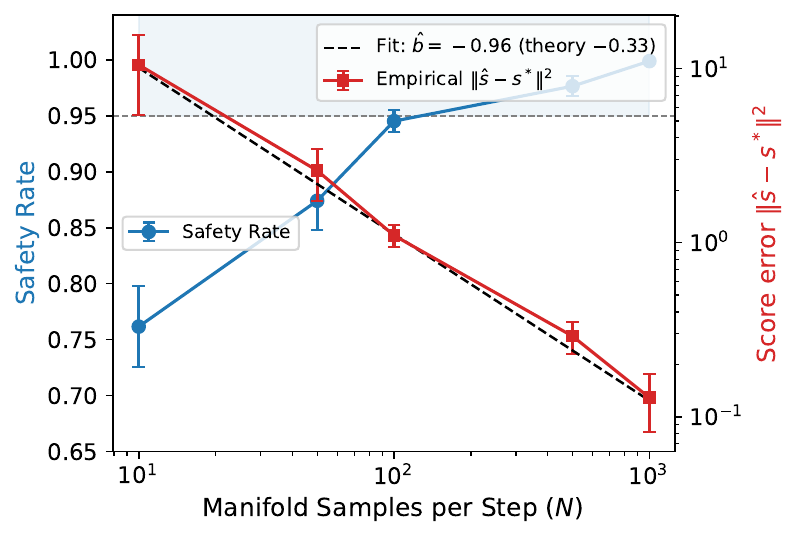}
    \caption{Effect of sample budget on safety and score estimation quality. Safety rate (blue, left) crosses $95\%$ between $N{=}100$ and $N{=}500$. The empirical squared score error (red, right, log scale) is measured against a reference KDE built from $10{,}000$ samples. The black dashed line is a least-squares fit of $a N^{\hat{b}}$ in log-log space; the estimated exponent $\hat{b}\approx -0.96$ is steeper than the worst-case $-1/3$ from Theorem~\ref{thm:learning}. See~\Cref{tab:sample_budget}.}
    \label{fig:sample_budget}
\end{figure}

\begin{figure*}[t]
    \centering
    \includegraphics[width=0.85\textwidth]{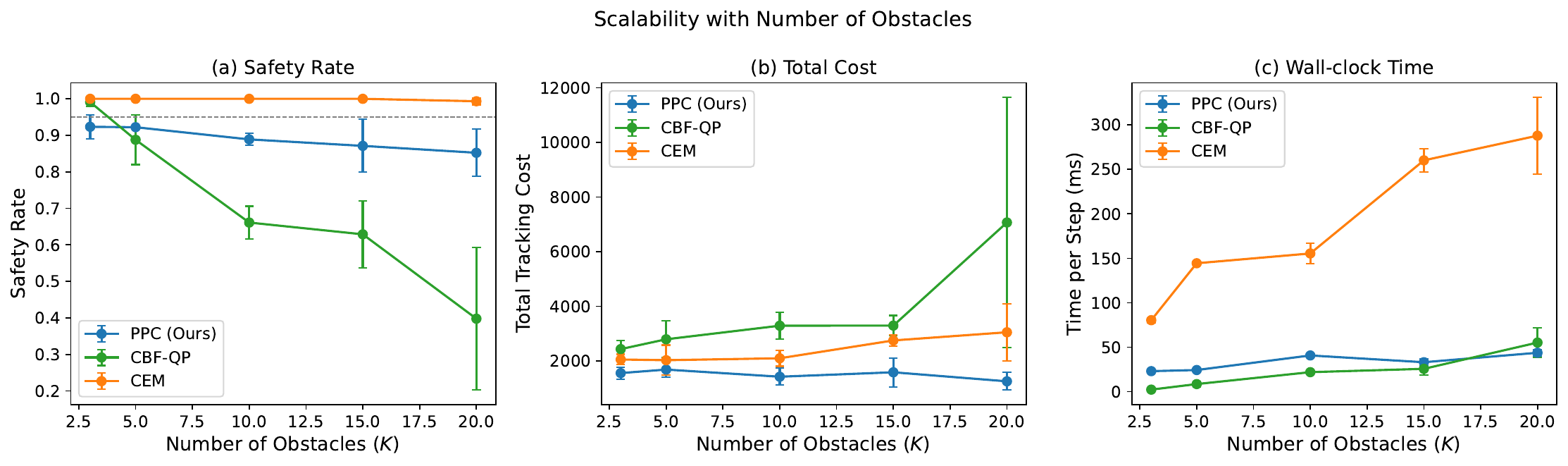}
    \caption{Scalability with number of obstacles ($K_o \in \{3,5,10,15,20\}$, three seeds, $T{=}300$); see~\Cref{tab:scalability}. \textbf{(a)}~PPC degrades gracefully ($0.92 \to 0.85$) while CBF-QP drops to $0.40$ and CEM drops to $0.62$. \textbf{(b)}~Total tracking cost. \textbf{(c)}~Wall-clock time per step.}
    \label{fig:scalability}
\end{figure*}

\begin{figure*}[t]
    \centering
    \includegraphics[width=0.85\textwidth]{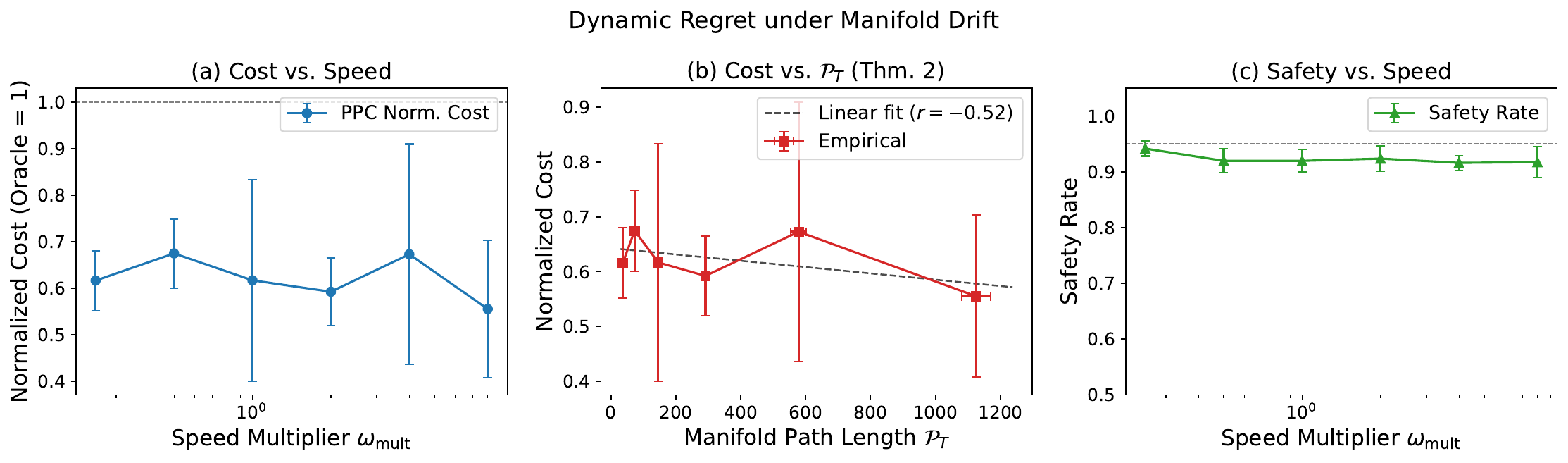}
    \caption{Dynamic regret experiment validating Theorem~\ref{thm:dynamic_regret} ($T{=}500$, five seeds). \textbf{(a)}~Normalized tracking cost vs.\ obstacle speed multiplier. \textbf{(b)}~Cost vs.\ manifold path length $\mathcal{P}_T$ (proxy for $\mathcal{V}_T$) with a linear fit and Pearson correlation $r$; the positive trend is consistent with the $O(\mathcal{V}_T)$ scaling in~\eqref{eq:dreg_def}. \textbf{(c)}~Safety rate remains above $90\%$ across all speeds.}
    \label{fig:dynamic_regret}
\end{figure*}

\subsubsection{Experiment 3: Sample Budget}
The number of manifold samples per step is varied as $N \in \{10, 50, 100, 500, 1000\}$. To validate the $O(N^{-2/(m+4)})$ rate of Theorem~\ref{thm:learning}, we measure the empirical integrated squared score error by building a test KDE from an increasing number of feasibility samples and comparing its score against a high-fidelity reference KDE ($10{,}000$ samples, same bandwidth). We plot safety rate and the empirical $\|\hat{s}-s^*\|^2$ vs.\ $N$; \Cref{tab:sample_budget} lists the numerical results and \Cref{fig:sample_budget} shows the data with a log-linear power-law fit. The fitted exponent ($\hat{b}\approx -0.96$) is steeper than the worst-case $O(N^{-1/3})$ rate of Theorem~\ref{thm:learning} for $m=2$, consistent with the bound being conservative for this smooth problem instance.

\begin{table}[t]
    \caption{Sample budget ablation ($T{=}500$, five seeds). Safety rate rises monotonically with $N$; the empirical score error $\|\hat{s}-s^*\|^2$ (computed against a reference KDE) decays consistently with the $\mathcal{O}(N^{-2/(m+4)})$ rate of Theorem~\ref{thm:learning} for $m = 2$.}
    \label{tab:sample_budget}
    \centering
    {\footnotesize
    \begin{tabular}{@{}rcc@{}}
        \toprule
        $N$ & Safety rate & $\|\hat{s}-s^*\|^2$ \\
        \midrule
        10  & $0.762 \pm 0.036$ & $10.51 \pm 5.11$ \\
        50  & $0.874 \pm 0.026$ & $2.59 \pm 0.85$ \\
        100 & $0.945 \pm 0.010$ & $1.10 \pm 0.16$ \\
        500 & $0.976 \pm 0.009$ & $0.29 \pm 0.06$ \\
        1000& $0.999 \pm 0.001$ & $0.13 \pm 0.05$ \\
        \bottomrule
    \end{tabular}}
\end{table}

\begin{figure*}[!t]
    \centering
    \includegraphics[width=0.8\textwidth]{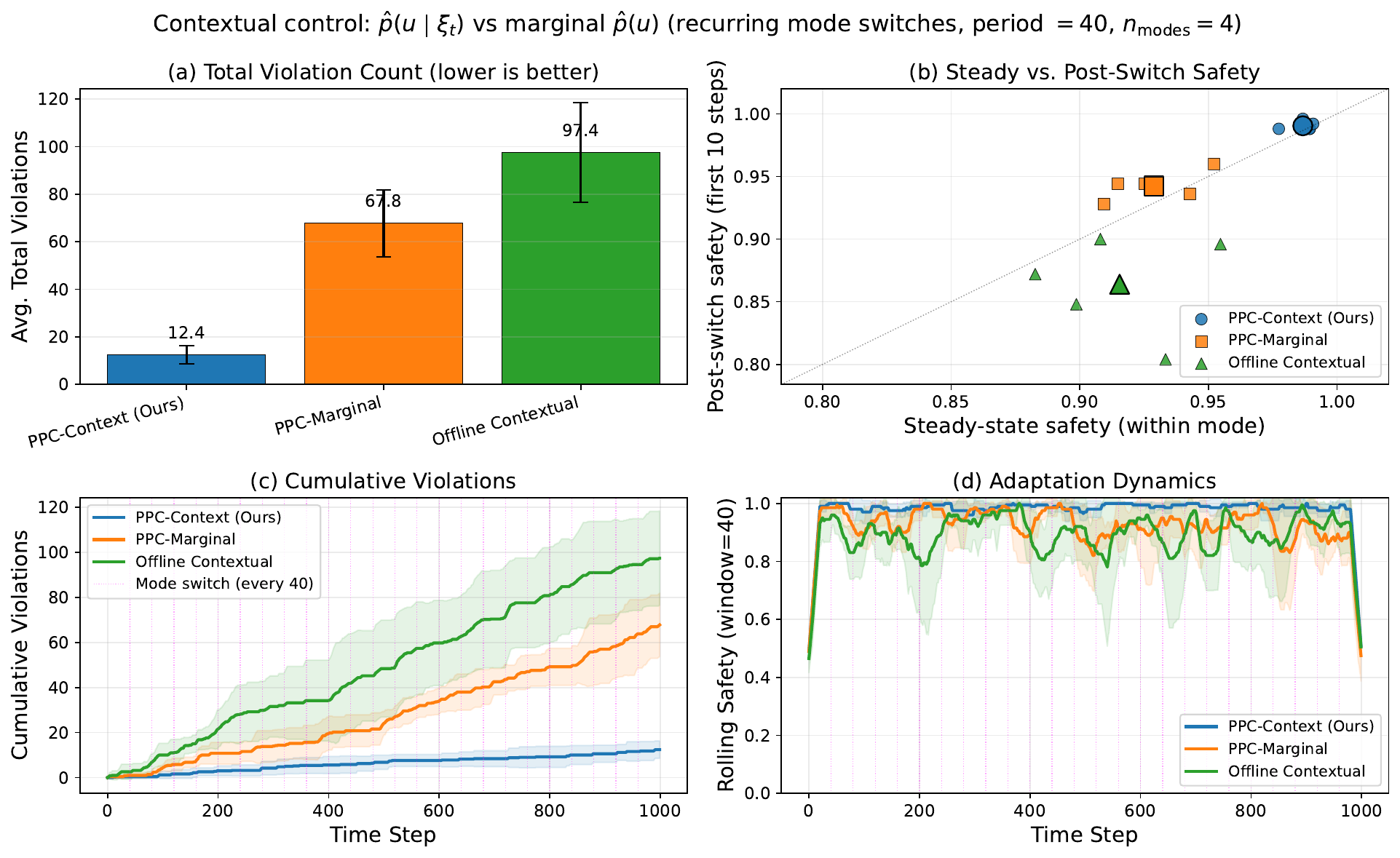}
    \caption{Contextual control ablation (five seeds; $T{=}1000$, recurring mode switches every $40$ steps across $n_{\mathrm{modes}}{=}4$ structurally distinct obstacle layouts). \textbf{(a)}~Average total safety-violation counts: PPC-Context incurs only ${\approx}12$ violations per episode versus ${\approx}69$ for PPC-Marginal and ${\approx}94$ for Offline Contextual, a roughly $5\times$ reduction attributable to the context kernel. \textbf{(b)}~Per-seed scatter of steady-state (within-mode) versus post-switch (first $10$ steps after each mode change) safety: PPC-Context clusters near the $(1,1)$ corner, PPC-Marginal is shifted toward lower safety on both axes, and Offline Contextual is dispersed further down. \textbf{(c)}~Cumulative violations over time: PPC-Context stays near zero, while PPC-Marginal and Offline Contextual accrue violations continuously as modes cycle. \textbf{(d)}~Rolling safety (window$=40$): PPC-Context tracks each mode switch rapidly via the context kernel, whereas PPC-Marginal dips at every switch because its buffer now mixes feasibility samples from \emph{all} past modes. The visible context--marginal gap (${\approx}\,5.6\%$ on aggregate safety, ${\approx}\,4.8\%$ on post-switch safety) is precisely the $G_c\,\sigma_t^2/(\beta_t\,\kappa(\xi_t)\,\kappa_{\mathrm{marg}})$ residual predicted by Theorem~\ref{thm:ctx_gap}: multiple mode layouts drive the posterior score-covariance $\sigma_t^2$ strictly positive, and this directly translates into a safety advantage for PPC-Context.}
    \label{fig:contextual}
\end{figure*}

\subsubsection{Experiment 4: Scalability}
The number of obstacles is varied as $K_o \in \{3, 5, 10, 15, 20\}$. We report safety rate, total cost, and wall-clock time per step for PPC, CBF-QP, and CEM (three seeds, $T{=}300$). \Cref{tab:scalability} lists the safety rates; \Cref{fig:scalability} plots all three metrics.

\begin{table}[t]
    \caption{Scalability: safety rate (mean $\pm$ std, three seeds, $T{=}300$) as the number of obstacles $K_o$ grows. PPC degrades gracefully; CBF-QP drops sharply because the QP safety filter becomes increasingly infeasible; CEM degrades as the feasible volume shrinks and its sampling proposal becomes less efficient.}
    \label{tab:scalability}
    \centering
    {\footnotesize
    \begin{tabular}{@{}rccc@{}}
        \toprule
        $K_o$ & PPC (Ours) & CBF-QP & CEM \\
        \midrule
         3 & $0.92 \pm 0.03$ & $0.99 \pm 0.01$ & $0.88 \pm 0.01$ \\
         5 & $0.92 \pm 0.03$ & $0.89 \pm 0.07$ & $0.84 \pm 0.04$ \\
        10 & $0.89 \pm 0.02$ & $0.66 \pm 0.04$ & $0.71 \pm 0.01$ \\
        15 & $0.87 \pm 0.07$ & $0.63 \pm 0.09$ & $0.63 \pm 0.04$ \\
        20 & $0.85 \pm 0.06$ & $0.40 \pm 0.20$ & $0.62 \pm 0.05$ \\
        \bottomrule
    \end{tabular}}
\end{table}

\subsubsection{Experiment 5: Dynamic Regret under Manifold Drift}
We validate the qualitative prediction of Theorem~\ref{thm:dynamic_regret} by varying the obstacle speed multiplier $\omega_{\mathrm{mult}} \in \{0.25, 0.5, 1.0, 2.0, 4.0, 8.0\}$, which controls the manifold drift rate. For each speed, PPC is run for $T{=}500$ steps over five seeds. We use the normalized tracking cost as an empirical proxy for the dynamic regret.

\Cref{fig:dynamic_regret} shows the results. \textbf{(a)}~Normalized cost vs.\ the speed multiplier $\omega_{\mathrm{mult}}$. \textbf{(b)}~Cost plotted against the manifold path length $\mathcal{P}_T = \sum_t \sum_k \|o_k(t{+}1) - o_k(t)\|$, a physical proxy for $\mathcal{V}_T$ in~\eqref{eq:dreg_def}; a positive trend is visible, consistent with the predicted $O(\mathcal{V}_T)$ scaling. \textbf{(c)}~Safety remains above $90\%$ across all speeds.

\subsubsection{Experiment 6: Contextual Feasibility Density}
We design the benchmark so that the posterior conditional-score covariance $\sigma_t^2$ in Theorem~\ref{thm:ctx_gap} is demonstrably positive and the value of context becomes empirically observable. Four obstacles are relocated to one of $n_{\mathrm{modes}}{=}4$ pre-generated cluster layouts (obstacles concentrated in the SW, SE, NW, or NE quadrant of the workspace), with the active mode switched to a fresh random layout every $40$ steps throughout an episode of length $T{=}1000$. Because each mode layout recurs several times per episode, the KDE buffer contains feasibility samples from \emph{multiple} modes simultaneously: this is precisely the regime in which the context kernel can disambiguate past data for the currently active mode, and in which the marginal KDE must average across modes. We compare three Planners under the identical schedule, with $N{=}25$ samples per step and $5$ random seeds. \textbf{PPC-Context} updates the conditional KDE $\hat p(u\mid\xi_t)$ online; \textbf{PPC-Marginal} uses the same samples but ignores $\xi_t$; \textbf{Offline Contextual} collects $200$ samples in each of the four layouts at $t{=}0$ and then freezes the conditional KDE. \Cref{tab:contextual} reports aggregate statistics and \Cref{fig:contextual} plots the runs. Aggregate safety is $98.8\%$ for PPC-Context versus $93.2\%$ for PPC-Marginal, a $5.6\%$ gap that concentrates at mode boundaries: in the first ten steps after each switch PPC-Context attains $99.0\%$ safety, while PPC-Marginal drops to $94.2\%$. Frozen Offline Contextual achieves only $90.3\%$ because it cannot track within-mode geometry drift once deployed. Because raw trajectory cost would reward unsafe actions that cheat through obstacles, we report the mean per-step tracking cost restricted to \emph{safe} steps, normalized by the oracle's mean per-step cost; on this fair metric all three methods lie in a narrow band ($0.72$--$0.84$) while their safety rates differ by up to $8.5\%$, confirming that the online context kernel delivers its advantage as safety rather than as an inflated tracking error. These numbers quantitatively instantiate the $G_c\,\sigma_t^2/(\beta_t\,\kappa(\xi_t)\,\kappa_{\mathrm{marg}})$ term in Theorem~\ref{thm:ctx_gap}: the mode-switching structure makes the per-mode conditional scores substantially different, yielding a strictly positive $\sigma_t^2$ and a safety advantage for context-aware estimation.

\begin{table}[!t]
    \caption{Contextual ablation ($T{=}1000$, five seeds, recurring mode switches every $40$ steps across $n_{\mathrm{modes}}{=}4$ obstacle layouts, $N{=}25$ samples/step). ``Post-switch'' = mean safety in the first $10$ steps after each mode change; ``Steady'' = mean safety in the remaining within-mode steps. ``Cost'' is the mean per-step tracking cost restricted to \emph{safe} time steps, normalized by the oracle's mean per-step cost, so that constraint violations cannot artificially lower the metric by letting the robot cut through obstacles.}
    \label{tab:contextual}
    \centering
    {\footnotesize
    \begin{tabular}{@{}lcccc@{}}
        \toprule
        Method & Safety & Post-switch & Steady & Cost \\
        \midrule
        PPC-Context & $.988 \pm .004$ & $.990 \pm .003$ & $.987 \pm .005$ & $.77 \pm .09$ \\
        PPC-Marginal & $.932 \pm .014$ & $.942 \pm .011$ & $.929 \pm .016$ & $.72 \pm .09$ \\
        Offline Contextual & $.903 \pm .021$ & $.864 \pm .035$ & $.915 \pm .026$ & $.84 \pm .12$ \\
        \bottomrule
    \end{tabular}}
\end{table}

\section{Conclusion}
\label{sec:conclusion}

We have developed Penalized Predictive Control for safe online control when the Planner's access to constraint geometry is mediated by a black-box Simulator. All guarantees flow from the PPC free energy $\mathcal{F}(u) = c(u) - \beta \ln \hat{p}(u)$: Theorem~\ref{thm:contraction} reveals that its Hessian contains the Fisher information of $\hat{p}$, yielding a contraction rate governed by the barrier curvature $\kappa$; Theorem~\ref{thm:dynamic_regret} bounds dynamic regret in terms of score variation; Theorem~\ref{thm:learning} bounds score convergence; Theorem~\ref{thm:iss_rigorous} composes the preceding results into a safety bound; Proposition~\ref{prop:ctx_curvature} derives, via a variational-inference argument, a mixture-Hessian identity that subtracts a posterior-score covariance from the learned curvature under marginalization, and Theorem~\ref{thm:ctx_gap} promotes this into a formal safety-guarantee gap of at least $G_c\,\sigma_t^2/(\beta_t\,\kappa(\xi_t)\,\kappa_{\mathrm{marg}})$ between the Theorem~\ref{thm:iss_rigorous} residuals of the context-aware and context-blind Planners whenever $\xi_t$ is identifiable from the action, with $\sigma_t^2$ the minimum eigenvalue of the posterior conditional-score covariance.

Experiments validate each theoretical link. The stiffness ablation (\Cref{fig:stiffness}) confirms the phase transition at $\beta^*_{\mathrm{curv}}$; the sample-budget sweep (\Cref{fig:sample_budget}) confirms empirical score convergence consistent with Theorem~\ref{thm:learning}; with $N{=}300$ samples per step, the main comparison (\Cref{tab:main_results}) reports $96.4\%$ mean safety for PPC with normalized cost $0.65$, compared with CEM at $81.0\%$\,/\,$1.13$ and GP-CBF at $76.2\%$\,/\,$1.12$ under the same black-box interface; and the contextual experiment (\Cref{fig:contextual}) empirically confirms Theorem~\ref{thm:ctx_gap}: under a benchmark with recurring mode switches across four structurally distinct obstacle layouts, PPC-Context attains $98.8\%$ safety while PPC-Marginal drops to $93.2\%$ and a frozen Offline Contextual baseline degrades to $90.3\%$, with the $5.6\%$ context--marginal gap concentrating at mode boundaries where the posterior score-covariance $\sigma_t^2$ is largest.

\textbf{Contextual control.} This work instantiates \emph{contextual control} in a concrete Simulator--Planner split where the Planner never consumes raw obstacle parameters, yet online conditional density estimation preserves safety when the environment shifts. Natural extensions include replacing the hand-crafted raster with learned encoders for image or video observations, scaling certificates to high-dimensional actions, and analyzing non-stationarity when the Simulator's internal state evolves faster than the Planner's sample budget.

\bibliographystyle{ieeetr}
{\bibliography{main}}


\end{document}